\documentclass[review]{elsarticle}

\usepackage{lineno,hyperref}
\modulolinenumbers[5]
\usepackage{amsthm}  
\usepackage{diagbox}
\usepackage{amsfonts,amssymb,amsbsy,amsmath}
\usepackage{lipsum}
\usepackage{graphicx}
\usepackage{epstopdf}
\usepackage{algorithmic}

\usepackage{graphicx}
\usepackage{epstopdf}
\usepackage{geometry}
\usepackage{epsfig}
\usepackage{bookmark}
\usepackage{multirow}
\usepackage{multirow}
\usepackage{algorithm}
\usepackage[figuresright]{rotating}
\usepackage{amscd}
\usepackage{mathrsfs}
\usepackage{booktabs,longtable}
\usepackage{indentfirst}
\usepackage{subfigure}
\usepackage{amstext}
\usepackage{fancyhdr}

\usepackage{xcolor}

\usepackage[nodots]{numcompress}

\usepackage[capitalize]{cleveref}

\usepackage{titlesec}
\usepackage{threeparttable}
\usepackage[title]{appendix}

\modulolinenumbers[5]
\biboptions{sort&compress}
\numberwithin{equation}{section}
\allowdisplaybreaks
\newtheorem{theorem}{Theorem}[section]

\newtheorem{corollary}{Corollary}[section]
\newtheorem{lemma}{Lemma}[section]
\newtheorem{remark}{Remark}[section]

\newtheorem{condition}{Condition}[section]

\journal{arXiv}

\begin{document}

\begin{frontmatter}

\title{Optimal Subsampling for High-dimensional Ridge Regression\tnoteref{mytitlenote}}
\tnotetext[mytitlenote]{The work was  supported by the National Natural Science Foundation of China (No. 11671060) and the Natural Science Foundation Project of CQ CSTC (No. cstc2019jcyj-msxmX0267).}

\author  {Hanyu Li \corref{mycorrespondingauthor}}
\cortext[mycorrespondingauthor]{Corresponding author}
\ead{lihy.hy@gmail.com or hyli@cqu.edu.cn}

	\author{Chengmei Niu}
\ead{chengmeiniu@cqu.edu.cn.}

\address{College of Mathematics and Statistics, Chongqing University, Chongqing 401331, P.R. China}

\begin{abstract}
We investigate the feature compression of  high-dimensional ridge regression using the optimal subsampling technique. Specifically, based on the basic framework of random sampling algorithm on feature for ridge regression and the A-optimal design criterion, we first obtain a set of optimal subsampling probabilities. 
Considering that the obtained probabilities are uneconomical, we then propose the nearly optimal ones. With these probabilities, a two step iterative algorithm is  established which has lower computational cost and higher accuracy. We provide theoretical analysis 
and numerical experiments 
to support the proposed methods. Numerical results demonstrate the decent performance of our methods.
\end{abstract}


\end{frontmatter}

\linenumbers

\section{Introduction}\label{sec.1}
For the famous linear model
$${y}={A}{\beta}+{\upsilon}, $$
where ${y} \in \mathbb{R}^{n}$ is the response vector,   ${{A}} \in \mathbb{R}^{n\times p}$ is the design matrix,   ${\beta}\in \mathbb{R}^{p}$ is the  parameter  vector, and  ${\upsilon} \in \mathbb{R}^{n}$  is the standardized Gaussian noise vector,
ridge regression \cite{hoerl1970ridge}, also known as
the least squares regression with Tikhonov regularization \cite{tihonov1963solution},
has the following form 
	\begin{align} \label{rls.1.0} 
		\mathop{\rm{min}}\limits_{{{\beta}}}\frac{1}{2}\|{y}-{A}{\beta}\|_2^2 + \frac{\lambda}{2} \| {\beta}\|_2^2,	
	\end{align}
	where   
	$\lambda$ is the regularized parameter,
	and the corresponding estimator is
\begin{align*}
	\hat{\beta}_{rls}=({A}^T{A}+\lambda {I})^{-1}{A}^T{y}.
\end{align*}
In this paper, we focus only on the case $p>n$, i.e., the high dimensional ridge regression.  For this case, the dominant computational cost of the above estimator is from the matrix inversion which takes $ O(p^3) $ flops.
A straightforward  way of amelioration  is to  solve the problem \eqref{rls.1.0}
in the dual space. Specifically, we first solve the dual problem of \eqref{rls.1.0},
\begin{align} \label{rloss.2}
		\mathop{\rm{min}}\limits_{z}\frac{1}{2\lambda}\|A^Tz\|_2^2+\frac{1}{2}\|z\|_2^2-z^Ty,
\end{align}
and the solution is
	\begin{align}
		\hat{z}^*&=\lambda(AA^T+\lambda I)^{-1}y. \label{rls.2}
	\end{align}
Then, setting
	\begin{align}
		 \hat{\beta}_{rls}&=\frac{A^T\hat{z}^*}{\lambda} \label{rls.2.0}
	\end{align}  
gives the estimator of \eqref{rls.1.0} in an alternative form
\begin{align} \label{rls.1.1}
	\hat{\beta}_{rls}=A^T(AA^T+\lambda I)^{-1}y.
\end{align}  
More details can be found in  \cite{saunders1998ridge}. Now, the dominant computational cost is  $O(n^2p)$ which appears in the computation of {$AA^T$}. However, it is still prohibitive when $p\gg n $.

To reduce the computational cost, some scholars considered the randomized sketching technique \cite{lu2013faster,chen2015fast,avron2016sharper,wang2017sketching,chowdhury2018iterative,lacotte2020adaptive}. 
The main idea is to compress the design matrix $ A $ to be a small  one $ \hat{A} $   by post-multiplying it by a random matrix $S\in \mathbb{R}^{p \times r}$ with $r\ll p$, i.e.,  $\hat{A}=AS$, and hence the reduced regression can be called the compressed ridge regression. 
There are two most common ways to generate $S$: random   projection and random sampling. The former can be the (sub)Gaussian matrix \cite{avron2016sharper,wang2017sketching,lacotte2020adaptive}, the sub-sampled randomized Hadamard transform (SRHT) \cite{lu2013faster,chen2015fast,avron2016sharper,wang2017sketching,lacotte2020adaptive}, the sub-sampled randomized Fourier transform \cite{wang2017sketching}, and the CountSketch (also called the sparse embedding matrix) \cite{avron2016sharper}, and the latter can be the uniform sampling and the  importance sampling \cite{chowdhury2018iterative}.

Specifically,  building on \eqref{rls.2} and \eqref{rls.2.0},
Lu et al. 
\cite{lu2013faster} presented the following  estimator \begin{align*}
	\hat{\beta}_{L}=\frac{SS^TA^T\widetilde{z}_L}{\lambda},	 
\end{align*}
where $S$ is the SRHT and	\begin{align}\label{duiou}
\widetilde{z}_L
=\lambda(ASS^TA^T+\lambda I)^{-1}y
\end{align}
is the solution to the dual problem of the following compressed ridge regression
\begin{align}\label{duiouw1}
    	\mathop{\rm{min}}\limits_{\beta_H}\frac{1}{2}\|y-AS\beta_H\|_2^2 + \frac{\lambda}{2} \| \beta_H\|_2^2,
\end{align}
and obtained a  risk bound. 
Soon afterwards, for $S$ generated by the product of sparse embedding matrix and SRHT, Chen et al.
\cite{chen2015fast} developed an estimator as follows:
\begin{align}\label{dls.3}
	\hat{\beta}_C=A^T(AS)^{\dagger T}(\lambda (AS)^{\dagger T}+AS)^{\dagger}y,
\end{align}
where $\dagger  $ denotes the Moore-Penrose inverse, and provided an estimation error bound and  a risk bound. 
Later,  
Avron et al. \cite{avron2016sharper}  proposed the  estimator $\hat{\beta}_A=A^T\hat{b}$, where
\begin{align*}
	\hat{b}=\mathop{\rm{argmin}}\limits_{b}\frac{1}{2}\|ASS^TA^Tb\|_2^2 -y^TAA^Tb+\frac{1}{2}\|y\|_2^2+ \frac{\lambda}{2} \| S^TA^Tb\|_2^2
\end{align*}
 with $S$ being the CountSketch,  SRHT, or Gaussian matrix.
The above problem is the sketch of the following regression problem
\begin{align*}
	\mathop{\rm{min}}\limits_{b}\frac{1}{2}\|AA^Tb\|_2^2 -y^TAA^Tb+\frac{1}{2}\|y\|_2^2+ \frac{\lambda}{2} \| A^Tb\|_2^2,
\end{align*}
which is transformed from 
\eqref{rls.1.0}.
Additionally, Wang et. al \cite{wang2017sketching}  
and Lacotte and Pilanci \cite{lacotte2020adaptive} applied 
the dual random projection 
proposed in \cite{zhang2013recovering,zhang2014random} to the high-dimensional ridge regression.
By the way, there are some works on compressed least squares regression \cite{maillard2009compressed, fard2012compressed, kaban2013new, kaban2014new, thanei2017random, slawski2017compressed, slawski2018principal, mor2019sketching}, which can be written in the following form
\begin{align}\label{sls.1}
	\hat{\alpha}_{ls}=\mathop{\rm{argmin}}\limits_{\alpha}\frac{1}{2}\|y-AS\alpha\|_2^2,
\end{align}
where $S$ is typically the (sub)Gaussian matrix.

To the best of our knowledge, there is few work of applying  random sampling to  high-dimensional ridge regression. We only found a work of 
\cite{chowdhury2018iterative}, which 
proposed an iterative algorithm 
by using the random sampling   with the column leverage scores or  ridge leverage scores as the sampling probabilities.
This algorithm can be viewed as an extension of the method  in \cite{avron2016sharper}. However, there are some works on  compressed least squares regression via random sampling. As far as we know,  Drineas et al.
\cite{drineas2012fast} first applied the random sampling with column leverage scores or approximated ones as  the sampling probabilities to the least squares
regression and established the following estimator
\begin{align*}
	\hat{\beta}_{D}=A^T(AS)^{\dagger T}(AS)^{\dagger}y,
\end{align*}
which can be regarded as a special case of \eqref{dls.3}. Later,  Slawski
\cite{slawski2018principal}  investigated \eqref{sls.1}  using uniform sampling, and discussed the predictive performance.

In this paper, we will consider the application of random sampling on  high-dimensional ridge regression further. Inspired by the  technique of the optimal subsampling used in e.g., \cite{Zhu2015, wang2018optimal, ma2020asymptotic, yao2019optimal, wang2021optimal, Zhang2021},  we will mainly investigate the optimal subsampling probabilities  for compressed ridge regression.
The nearly optimal subsampling probabilities and a two step iterative algorithm are also derived.

The remainder of this paper is organized as follows. 
The basic framework of random sampling algorithm and  the optimal subsampling  probabilities  are presented   in Section \ref{sec.3}.  In Section \ref{sec.4}, we propose the nearly optimal subsampling probabilities and  a two step iterative algorithm. The  detailed theoretical analyses of the proposed methods are also presented in Sections \ref{sec.3} and \ref{sec.4}, respectively.  In Section \ref{sec.5}, we provide some numerical experiments to test our methods. The proofs of all the main theorems 
are given in the appendix. 

Before moving to the next section, we introduce some standard notations used in this paper. 

For the matrix $A \in \mathbb{R}^{n\times p}$, 
$A_i $, $A^j$,   $\|A\|_2$ and $\|A\|_F$ denote its 
$ i $-th column, $ j $-th row, spectral norm and  Frobenius norm, respectively.   
Also,  its thin SVD is given as $A=U \Sigma V^T$,
where $U\in \mathbb{R}^{n\times \rho}$,  $V\in \mathbb{R}^{p\times \rho}$, and 
$ \Sigma\in \mathbb{R}^{ \rho \times \rho} $ with the diagonal elements, i.e., the singular values of $A$,  satisfying $\sigma_{1}(A)\geq\sigma_2(A)\geq\cdots\geq\sigma_{\rho}(A)>0$. 

For $V$,   its row norms $\|V^i\|_2$ with $i=1,\cdots,p$ are the  column leverage scores \cite{chowdhury2018iterative}, and  for $X=V\Sigma_{\lambda}$, where $ \Sigma_{\lambda} $ is a diagonal matrix  with the diagonal entries being $\sqrt{\frac{\sigma_j(A)^2}{\lambda+\sigma_j(A)^2}}$ ($ j=1,\cdots,\rho $),  its row norms $\|X^i\|_2$ are called the ridge leverage scores \cite{chowdhury2018iterative}.

In addition, $O_p(1)$ denotes that a sequence of random variables  are bounded in probability and $o_p(1)$ represents that the sequence convergences to zero in probability. More details  can refer to \cite[Chap.~2]{van2000asymptotic}. 	 In our case, we also use  $O_{p \mid \mathcal{F}_n}$  to denote that a sequence of random variables  are bounded in conditional probability given  the  full data matrix ${\mathcal{F}}_n=(A,y)$.  Especially, 
for any matrix $G$, $G=O_p(1)$ ($G=O_{p \mid \mathcal{F}_n}(1)$) means that all the elements of $G$ are bounded in probability (given ${\mathcal{F}}_n$),  and $G=o_p(1)$ symbolizes that  its elements are  convergence to zero in probability.

\section{Optimal Subsampling}
\label{sec.3}

In this section,  we 
will present the basic framework of random sampling algorithm, propose the 
optimal subsampling probabilities,  and obtain the corresponding  error analysis.

\subsection{ Algorithm and Optimal Subsampling Probabilities}
\label{sec.3.2}	
 Given a set of probabilities, i.e., the random sampling matrix $S$, our approximate estimator
\begin{align}\label{dsolu.3}
	\hat{\beta}=A^T(ASS^TA^T+\lambda I)^{-1}y
\end{align}
of the high-dimensional ridge regression \eqref{rls.1.0} is the combination of the solution to the compressed dual problem, 
\begin{align}\label{duiouw}
    \mathop{\rm{argmin}}\limits_{z}\frac{1}{2\lambda}\|S^TA^Tz\|_2^2+\frac{1}{2}\|z\|_2^2-z^Ty,
\end{align}
i.e., 
\begin{align}\label{duiou1}
\hat{z}
=\lambda(ASS^TA^T+\lambda I)^{-1}y,
\end{align}
and \eqref{rls.2.0}.  
That is, 
we first solve the problem \eqref{duiouw} and then get the approximate estimator through \eqref{rls.2.0}\footnote{Note that this approach is different from the one in \cite{lu2013faster} though the expressions of $\hat z$ in \eqref{duiou1} and $\widetilde{z}_L$ in \eqref{duiou} are the same. In fact, the authors in \cite{lu2013faster} first solve the compressed ridge regression \eqref{duiouw1} in the dual space and then  find the estimator of the compressed regression 
via \eqref{rls.2.0}. Finally, the approximate estimator of the original ridge regression is recovered by the random matrix $S$. }.
The detailed process, i.e., 
the basic framework of random sampling algorithm, is listed in Algorithm \ref{DSP}.

\begin{algorithm}[htbp]
	\caption{Random Sampling Algorithm for High-dimensional Ridge Regression (RSHRR) }\label{DSP}\small
	$\textbf{Input:}$ $y \in \mathbb{R}^{n}$, $A \in \mathbb{R}^{n\times p} $,  the regularized parameter $\lambda$,  the sampling size $r$  with $r\ll p$, and the sampling probabilities $\{\pi_{i}\}^{p}_{i=1}$ with $\pi_{i}\geq0$  such that $\sum_{i = 1}^{p}\pi_{i}=1$.\\
	$\textbf{Output:}$ the dual solution $\hat{z}$ and the primal solution $\hat{\beta}$.
	\begin{enumerate}
				\item initialize $S\in \mathbb{R}^{p\times r}$ to  an all-zeros  matrix.
		\item for   $t\in{1,\cdots,r}$ do 
		\begin{itemize}
			\item pick $i_t\in [p]$ such that $\mathrm{Pr}\mathit{(i_t=i)}=\pi_i $.
			\item set $S_{i_tt}=\frac{1}{\sqrt{r\pi_{i_t}}}$.
		\end{itemize}	
		\item end
		\item calculate $\hat{z}$ as in \eqref{duiou1}.
		\item return   $\hat{\beta}=\frac{A^T\hat{z}}{\lambda}$.	
	\end{enumerate}
\end{algorithm}

\begin{remark}\label{rem4.8}
{	In Algorithm \ref{DSP},  the parameter $\lambda$ 
	can be determined by $K$-fold  cross-validation,  leave-one-out
cross-validation, or generalized cross-validation, see e.g. \cite{chen2022optimal}.
Since the main focus of this paper is the performance of subsampling on high-dimensional ridge regression, we omit the investigation of the choice of $\lambda$.}
\end{remark}

Now, we investigate the  sampling probabilities $\{\pi_{i}\}^{p}_{i=1}$ in Algorithm \ref{DSP}, which play a critical role on the performance of the algorithm. Below are some well known  probabilities  discussed in the literature.

\begin{itemize}
	\item $\textbf{Uniform sampling (UNI)}$:  $\pi^{UNI}_i=\frac{1}{p}$.
	\item $\textbf{Column sampling (COL)}$: $\pi^{COL}_i=\frac{\|A_i\|^2_2}{\sum_{i=1}^{p}\|A_i\|^2_2}$.	
	\item $\textbf{Leverage sampling (LEV)}$ \cite{chowdhury2018iterative}:  $\pi^{LEV}_i=\frac{\|V^i\|^2_2}{\sum_{i=1}^{p}\|V^i\|^2_2}$.
	\item $\textbf{Ridge leverage sampling (RLEV)}$\cite{chowdhury2018iterative}: $\pi^{RLEV}_i=\frac{\|X^i\|^2_2}{\sum_{i=1}^{p}\|X^i\|^2_2}$.
	
\end{itemize}

In the following, we discuss a  new set of  sampling probabilities, i.e., the optimal subsampling probabilities, which can be derived by combining the asymptotic variance of the estimators  from Algorithm \ref{DSP} and  the A-optimal design criterion \cite{Pukelsheim1993}. 
Considering the property of trace 
\cite[Section 7.7]{horn2012matrix} 
and the variance   	
 $ \rm{Var}(\hat{\beta}-\hat{\beta}_{\mathit{rls}}|\mathcal{F}_n)=\frac{1}{\mathit{\lambda}^2}\mathit{A}^T \rm{Var}(\hat{\mathit{z}}-\hat{\mathit{z}}^{\ast}|\mathcal{F}_n)\mathit{A}$, 
 to let the trace ${\rm tr}(\rm{Var}(\hat{\beta}-\hat{\beta}_{\mathit{rls}}|\mathcal{F}_n))$  attain its minimum, it suffices to make
${\rm tr}(\rm{Var}(\hat{\mathit{z}}-\hat{\mathit{z}}^{\ast}|\mathcal{F}_n))$ get its minimum. Thus, 
we mainly investigate the asymptotic variance of the dual estimator  $ \hat{z} $. 
As done  in e.g., \cite{Zhu2015, wang2018optimal, yao2019optimal, wang2021optimal, Zhang2021},   
several conditions  are first presented as follows.
\begin{condition}\label{C.3}
	For the design matrix $A \in \mathbb{R}^{n\times p}$, 
	we assume that
	\begin{align}
		\sum_{i = 1}^{p}\frac{\|A_i\|^{6}_2}{\pi^{2}_ip^{3}}={O}_p(1),\label{cond.3}\\
		\mathit{\sum_{i=1}^{p}\frac{A_iA^T_i\|A_i\|\rm{^2_2}}{ p\rm{^2}\pi_\mathit{i}} =O_p(\rm{1})},\label{cond.1}\\
		\mathit	{\sum_{i=\rm{1}}^{p}\frac{\|A_i\|\rm{^2_2}}{p}=O_p(\rm{1})},\label{rem.c2.2}\\
		\sum_{i=1}^{p}\frac{A_iA^T_i}{ p}=O_p(1),\label{rem.c2.1}
	\end{align}
	where  $\pi_i$ with $i=1,\cdots,p$ are the given probabilities. 
\end{condition}	

\begin{remark}\label{rem.c.1.0}
		With respect to uniform sampling, i.e. $ \pi_i={p}^{-1} $,
	the   conditions \eqref{cond.3} and \eqref{cond.1} are equivalent to
	\begin{align}\label{cond.rem.1}
		\sum_{i = 1}^{p}\frac{\|A_i\|^{6}_2}{p}={O}_p(1), \quad \mathit{\sum_{i=1}^{p}\frac{A_iA^T_i\|A_i\|\rm{^2_2}}{p} =O_p(\rm{1})}.
	\end{align} 
	In this case, 
	to make \eqref{cond.rem.1} hold,    it is sufficient to
	suppose that $\mathrm{E}(\|A_i\|^6_2)<\infty$. 	
	Furthermore, the  conditions \eqref{rem.c2.2} and \eqref{rem.c2.1}  hold if  $\mathrm{E}(\|A_i\|^2_2)<\infty$.
\end{remark}

\begin{remark}\label{rem.c.1}
	The above   moment  type conditions are wild. For example, if the  entries of $A$ obey the sub-Gaussian distribution \cite{buldygin1980sub}, then all the conditions mentioned above are satisfied.  The reason is that the sub-Gaussian distribution owns finite moments up to any finite order.
\end{remark}	

With the above conditions, we can  present the following asymptotic distribution theorem.	

\begin{theorem}\label{thm.1}
	Assume that the conditions \eqref{cond.3}, \eqref{cond.1},  \eqref{rem.c2.2},  and  \eqref{rem.c2.1}
	are satisfied.
	Then, as $p \to \infty $, $  r \to \infty $,	conditional on $\mathcal{F}_n$ in probability, the estimator $ \hat{z} $ constructed by  Algorithm \ref{DSP} satisfies
	\begin{align}\label{thm1.1}
		V^{-1/2}(\hat{z}-\hat{z}^{\ast})\xrightarrow{L}N(0,I),
	\end{align}
	where the notation $\xrightarrow{L}$ represents the convergence in distribution, and
	\begin{align*}
		V=(\frac{M_A}{p})^{-1}\frac{V_c}{r}(\frac{M_A}{p})^{-1}
	\end{align*}   with    $M_A=AA^T+\lambda I$ and 
	$ 	V_c=\sum_{i=1}^{p}\frac{A_iA^T_i\hat{z}^{\ast}\hat{z}^{\ast T}A_iA^T_i}{p^2\pi_i} $.
	
\end{theorem}

Following the A-optimal design criterion and the asymptotic variance $V$  in \eqref{thm1.1}, we can provide the   optimal subsampling probabilities for Algorithm \ref{DSP} by minimizing  the trace $\rm{tr}(\mathit{V})$.  Noting that $M_A$ does not depend on $\pi_i$  and is nonnegative definite,  we get that  $V_c(\pi_1)\preccurlyeq V_c(\pi_2)$ is equivalent to $V(\pi_1)\preccurlyeq V(\pi_2)$ for any  two sampling probability sets $\pi_1=\{\pi^{(1)}_i\}^p_{i=1}$ and $\pi_2=\{\pi^{(2)}_i\}^p_{i=1}$. Thus, we can simplify the optimal criterion by avoiding computing $M^{-1}_A$, namely, we can calculate the  optimal subsampling probabilities  by minimizing   $\rm{tr}(\mathit{V}_\mathit{c})$ instead of  $\rm{tr}(\mathit{V})$. Actually, this can  be  viewed as the L-optimal design  criterion \cite{Pukelsheim1993} with $L=rp^{-2}M_A^2$.

\begin{theorem}	\label{thm.2}
	For Algorithm \ref{DSP}, when 
	\begin{align}\label{thm2.1}
		\pi^{OPL}_i=\frac{	\mid \hat{\beta}_{rls(i)}	\mid \|A_i\|_2}{\sum_{i=1}^{p}	\mid \hat{\beta}_{rls(i)}	\mid \|A_i\|_2},  \quad  {i=1,\cdots,p},
	\end{align}
	where $\hat{\beta}_{rls(i)}$	is the $i$-th element of the   ridge estimator $\hat{\beta}_{rls}$,  $\rm{tr}(\mathit{V}_\mathit{c})$ achieves its minimum.
	
\end{theorem}

\begin{remark}\label{remth2.2}
	When $ \lambda	\to 0^+$,
	\eqref{thm2.1} can be degraded to 
		the optimal subsampling probabilities of the compressed least squares regression. 
\end{remark}

\begin{remark}\label{rem.1}
Note that $V_c
=\lambda^2\sum_{i=1}^{p}\frac{\beta^2_{rls}A_{i}A_{i}^T}{p^2\pi_{i}}$. Thus, by
	\begin{align*}
		\mid \hat{\beta}_{rls}(i)	\mid =\|A^T_i(AA^T+\lambda I)^{-1}y\|_2\leq\|A_i\|_2\|(AA^T+\lambda I)^{-1}y\|_2,
	\end{align*}
	we have
	\begin{align*}
	\rm{tr}(\mathit{V}_\mathit{c})&\leq	\frac{\lambda^2\|(AA^T+\lambda I)^{-1}y\|^2_2}{p^2}\sum_{i=1}^{p}\frac{\|A_i\|^4_2}{\pi_i}\\
		&=\frac{\lambda^2\|(AA^T+\lambda I)^{-1}y\|^2_2}{p^2}\sum_{i=1}^{p}\pi_i\sum_{i=1}^{p}\frac{\|A_i\|^4_2}{\pi_i}.
	\end{align*}
	Further, by Cauchy-Schwarz inequality, we obtain
	\begin{align*}
		\frac{\lambda^2\|(AA^T+\lambda I)^{-1}y\|^2_2}{p^2}\sum_{i=1}^{p}\pi_i\sum_{i=1}^{p}\frac{\|A_i\|^4_2}{\pi_i}\geq\frac{\lambda^2\|(AA^T+\lambda I)^{-1}y\|_2}{p^2}(\sum_{i=1}^{p}\|A_i\|^2_2)^2.
	\end{align*}
	Thus,	analogous to Theorem \ref{thm.2}, we get that when 
	\begin{align}\label{proofrem1.3}
		\pi_i=\pi^{COL}_i=\frac{\|A_i\|^2_2}{\sum_{i=1}^{p}\|A_i\|^2_2},
	\end{align}
	the  upper bound of $	\rm{tr}(\mathit{V}_\mathit{c})$, i.e., $\frac{\lambda^2\|(AA^T+\lambda I)^{-1}y\|^2_2}{p^2}\sum_{i=1}^{p}\frac{\|A_i\|^4_2}{\pi_i}$, reaches  the minimum.   Obviously,
	\eqref{proofrem1.3} is   easier  to compute compared with \eqref{thm2.1}. However, we has to lose some accuracy as expense in this case.
	
	Similarly,
	based on $ \|A_i\|^2_2 \leq\|A\|^2_F$, we have
	\begin{align*}
		{\rm{tr}}(\mathit{V}_\mathit{c})\leq	\frac{\lambda^2\|A\|^2_F}{p^2}\sum_{i=1}^{p}\frac{\hat{\beta}^2_{rls(i)}}{\pi_i}=\frac{\lambda^2\|A\|^2_F}{p^2}\sum_{i=1}^{p}\pi_i\sum_{i=1}^{p}\frac{\hat{\beta}^2_{rls(i)}}{\pi_i}
	\end{align*}
	and 
	\begin{align*}
		\frac{\lambda^2\|A\|^2_F}{p^2}\sum_{i=1}^{p}\pi_i\sum_{i=1}^{p}\frac{\hat{\beta}^2_{rls(i)}}{\pi_i}\geq\frac{\lambda^2\|A\|^2_F}{p^2}(\sum_{i=1}^{p}	\mid \hat{\beta}_{rls(i)}	\mid )^2.
	\end{align*}
	Then, we find that when 
	\begin{align*}
		\pi_i=\pi^{RSIS}_i=\frac{	\mid \hat{\beta}_{rls(i)}	\mid }{\sum_{i=1}^{p}	\mid \hat{\beta}_{rls(i)}	\mid },
	\end{align*}
	the above upper bound of $	\rm{tr}(\mathit{V}_\mathit{c})$ reaches  the minimum.   Surprisingly, 
	$\pi^{RSIS}_i$ 
	corresponds to the screening  criteria of iteratively thresholded ridge regression screener given in \cite{fan2008sure}. This fact implies that the screener with the probabilities in \eqref{thm2.1} may perform better than the one in \cite{fan2008sure}.
\end{remark}	

\subsection{Error Analysis for RSHRR}
\label{sec.3.3}	

We first give an estimation error bound.

\begin{theorem}\label{thm.3}
	Assume that 
	\begin{align}\label{lem.5.1}
		c_1\|{{V}}^i\|_2\leq\|{{A}}_i\|_2\leq c_2\|{{V}}^i\|_2 \ {and} \  	{s_{{1}}\|{{V}}^i\|_{{2}}\|{{y}}\|_2\leq	\mid \hat{\beta}_{rls(i)}}	\mid \leq s_{{2}}\|{{V}}^i\|_{{2}}\|{{y}}\|_2,\quad {i=1,\cdots,p},
	\end{align}	where  $0< c_1\leq c_2$ and $0< s_1\leq s_2$,  and  let $r\geq   \frac{32{s_2c_2\rho}}{3s_1c_1{\epsilon}^2}\mathrm{ln}(\frac{4\rho}{\delta})$ with    $\epsilon,\delta \in (0,1)$. 
	Then, 
	for $S$ formed 
	by $\pi_i$  = $\pi^{OPL}_{i}$ and any $\epsilon$,  with the probability at least $1-\delta$, $\hat{\beta}$  constructed by Algorithm \ref{DSP} satisfies
	\begin{align}\label{thm3.0}
		\|\hat{\beta}-\hat{\beta}_{rls}\|_2\leq\epsilon\|\hat{\beta}_{rls}\|_2,
	\end{align}
	where $\hat{\beta}_{rls}$ is as in \eqref{rls.1.1}.
	
\end{theorem}

\begin{remark}\label{rem.3.1}
	The assumptions  in \eqref{lem.5.1}
	are reasonable  and  reachable due to  $A_i=U\Sigma(V^i)^T$ and 
	\begin{align*}
		\hat{\beta}_{rls(i)}=A^T_i(AA^T+\lambda I)^{-1}y=V^i(\Sigma+\lambda \Sigma^{-1})^{-1}U^Ty.
	\end{align*}
	{In fact,  for the worst case, $c_1=\sigma_1({A})$, $c_2=\sigma_1({A})$, and $s_1$ and $s_2$ are controlled by $ \mathop{\rm{min}}\limits_{\mathit{j}=1,\cdots,\rho}\{\frac{\sigma_\mathit{j}(\mathit{A})}{\sigma^2_\mathit{j}(\mathit{A})+\lambda}\} $ and $ \mathop{\rm{max}}\limits_{\mathit{j}=1,\cdots,\rho}\{\frac{\sigma_\mathit{j}(\mathit{A})}{\sigma^2_\mathit{j}(\mathit{A})+\lambda}\} $, respectively.   } The aim for introducing the parameters $c_1, c_2, s_1,$ and $s_2$ here is to simplify the expression of $r$.
\end{remark}

In the following, we provide a risk  bound, in which the risk function is defined as
\begin{align*}
	\rm{risk}(\hat{\mathit{y}})=\frac{1}{\mathit{n}}\mathrm{E}_\mathit{y}(\mathit{\|\hat{y}-A\beta\|}\rm{^2_2}),
\end{align*}
where $\mathrm{E}_y$ denotes the expectation on $y$, and $\hat{y}$ denotes the  prediction of $A\beta$.

\begin{theorem}\label{thm.4}
	Suppose that the setting is the same as the one in
	Theorem \ref{thm.3},  and let  $\mu=\sqrt{\sum_{j=1}^{\rho}\frac{\sigma^2_j(A)}{(\sigma^2_j(A)+\lambda)^2}}$. 
	Then,  for $S$ formed 
	by $\pi_i$  = $\pi^{OPL}_{i}$ and any $\epsilon$,	with probability at least $1-\delta$, 
	\begin{align*} 
		\rm{risk}(\hat{\mathit{y}})\leq \rm{risk}({\mathit{y}}_{\ast})+\frac{3\mathit{\epsilon}}{\mathit{n}}\|\mathit{A}\|^2_2(\mathit{\mu}^2+\|\mathit{\beta}\|^2_2),
	\end{align*}
	where $\hat{y}=A\hat{\beta}$ with $ \hat{\beta}  $ constructed by Algorithm \ref{DSP} and ${y}_{\ast}=A{\hat{\beta}_{rls}}$.
	
\end{theorem}

\section{Two Step Iterative Algorithm}
\label{sec.4}

Considering that the sampling probabilities \eqref{thm2.1} are
uneconomic since they are required to figure out  $\hat{\beta}_{rls}$, we now present the approximate ones. Specifically, we first
apply Algorithm \ref{DSP} with $\pi_i=\pi^{COL}_i$   and the sampling size being $r_0$ to return an approximation  $\widetilde{\beta}$ of $\hat{\beta}_{rls}$. Then, a set
of probabilities $\{\pi^{NOPL}_i\}_{i=1}^p $ 
are obtained by replacing $\hat{\beta}_{rls(i)}$ in \eqref{thm2.1} with $\widetilde{\beta}_{(i)}$, i.e.,  
\begin{align}\label{sec4.1.0}	
	\pi^{NOPL}_i=\frac{	\mid \widetilde{\beta}_{(i)}	\mid \|A_i\|_2}{\sum_{i=1}^{p}	\mid \widetilde{\beta}_{(i)}	\mid \|A_i\|_2},\quad i=1,\cdots,p.
\end{align} 
We call them the nearly optimal subsampling probabilities. Moreover, to further reduce the estimation error, we bring in the iterative method. The key motivation is that  if 
$\|\hat{\beta}_t-\hat{\beta}_{rls}\|_2\leq\epsilon\|\hat{\beta}_{t-1}-\hat{\beta}_{rls}\|_2$ holds at the $t$-th iteration, 
then a solution owning  the estimation error bound $\epsilon^m\|\hat{\beta}_{0}-\hat{\beta}_{rls}\|_2$ will be returned when the approximation process is repeated 
$m$ times. Putting the above discussions together, we propose a two step iterative algorithm, i.e., Algorithm \ref{IDSP}.

\begin{algorithm}[htbp]
	\caption{Two Step Iterative  Algorithm for High-dimensional Ridge Regression} \label{IDSP}\small
	$\textbf{Input:}$ $y \in \mathbb{R}^{n}$, $A \in \mathbb{R}^{n\times p} $,  the regularized parameter $\lambda$, the iterative number $m$,  the sampling size $r$ and $r_0$, where $r_0\ll r\ll p$.\\
	$\textbf{Output:}$ the dual estimator 	$\hat{z}_m$ and the recovered solution $\hat{\beta}_m$.\\
	$\textbf{Step1:}$
	\begin{enumerate}
			\item initialize $S^{\ast}\in \mathbb{R}^{p\times r_0}$ to  an all-zeros  matrix.
		\item for   $i\in{1,\cdots,p}$ do 
		\begin{itemize}
			\item
			$  \pi^{COL}_i=\frac{\|A_i\|^2_2}{\sum_{i=1}^{p}\|A_i\|^2_2} $.
		\end{itemize}
		
		\item end
		\item for   $t\in{1,\cdots,r_0}$ do 
		\begin{itemize}
			\item 	pick $i_t\in [p]$ such that $\mathrm{Pr(\mathit{i_t=i)}}=\pi_i $.
			\item 
			$S^{\ast}_{i_tt}=\frac{1}{\sqrt{r_0\pi_{i_t}}}$.
		\end{itemize}

		\item end
		\item  compute $A^{\ast}=AS^{\ast}$.	
		\item compute $C=(A^{\ast}A^{\ast T}+\lambda I)^{-1}$.
	\end{enumerate}
	$\textbf{Step2:}$
	\begin{enumerate}
		\item set $\hat{z}_0=0$.  
		\item for   $t\in{1,\cdots,m}$ do
		\begin{itemize}
			
			\item
			$ \hat{\beta}_{t-1}=\frac{1}{\lambda}A^T\hat{z}_{t-1}$.	
			\item$ 
			b_t=y-A\hat{\beta}_{t-1}-\hat{z}_{t-1}$.
			\item$ 
			\widetilde{z}=\lambda C b_t$.
			\item
			$ \widetilde{\beta}=\frac{A^T\widetilde{z}}{\lambda}$. 
			\item compute
			$
			\pi^{NOPL}_i$ 
			by \eqref{sec4.1.0}.
			\item
			compute  $ \hat{w}_t $ by applying Algorithm \ref{DSP} with $y=b_t$ 
			and $\pi=\pi^{NOPL}_i$. 
			\item $\hat{z}_t=\hat{z}_{t-1}+\hat{w}_t$.
		\end{itemize}
		\item end  
		\item return
		$\hat{z}_m$ and $\hat{\beta}_m=\frac{A^T\hat{z}_m}{\lambda}$.	
	\end{enumerate}	
\end{algorithm}

\begin{remark}\label{rem4.3}
	The step 2 of Algorithm \ref{IDSP} can be viewed as a  variant of iterative Hessian sketch (IHS) \cite{wang2017sketching}. This  is because, at the $t$-th iteration, applying Algorithm \ref{DSP} for finding $\hat{w}_t$ is equivalent  to applying Hessian sketch to  the  residual between $z$ and  $\hat{z}_{t-1}$. 
	That is, at  the $t$-th iteration, we need to solve the following problem 
		\begin{align*}
			\mathop{\rm{min}}\limits_{w_t}\frac{1}{2\lambda}\|{S}^TA^Tw_t\|_2^2+\frac{1}{2}\|w_t\|_2^2-w_t^Tb_t,
		\end{align*} 
		where  
		$ w_t=z-z_{t-1} $
		and ${S}$ is constructed by $\pi^{NOPL}_i$. 
	
	In addition,  the step 2 of Algorithm \ref{IDSP} is also similar to   Algorithm 1 in \cite{chowdhury2018iterative}. However, the key ideas of the two methods are different. The latter can be regraded as  the preconditioned 
	Richardson iteration \cite[Chap.~2]{quarteroni2008numerical} for  solving $(AA^T+\lambda I)z=\lambda y$ with pre-conditioner $P^{-1}=(ASS^TA^T+\lambda I)^{-1}$ and the step-size being one. Moreover, its  random sampling matrix $S$ is fixed during the iteration.
\end{remark}

Next,  we show that the difference of $\hat{z}_{\ast}$ and  $\hat{z}_1$ still obeys asymptotically normal distribution, where $\hat{z}_1$ is  returned from Algorithm \ref{IDSP} with $m=1$. 	

\begin{theorem}\label{thm.5}
	Suppose that the conditions \eqref{rem.c2.2} and \eqref{rem.c2.1}   hold, and let 
	\begin{align}\label{thm5.0}
			N_1\|A_i\|_2\|{{y}}\|_2\leq	\mid \widetilde{\beta}_{(i)}\mid  \leq N_2\|A_i\|_2 \|{\mathit{y}}\|_2  \ {and} \ 
			{N}_3\|{A_i}\|_2\|{{y}}\|_2\leq	\mid {\hat{\beta}_{rls(i)}}	\mid \leq {N}_4\|{A_i}\|_2\|{{y}}\|_2, \quad i=1,\cdots,p, 
	\end{align}
	where $\widetilde{\beta}_{(i)}$ is	 as in Algorithm \ref{IDSP}, $0< N_1\leq N_2$, and $0<N_3\leq N_4$. 
	Then, as $ p	\to \infty $,  $ {r}	\to \infty $, $ r_0	\to \infty $,
	conditional on $\mathcal{F}_n$ and $\widetilde{\beta}$ in probability,
	the dual estimator $\hat{z}_1$ constructed by Algorithm \ref{IDSP} satisfies
	\begin{align}\label{thm5.1}
		V_{OPL}^{-1/2}(\hat{z}_1-\hat{z}^{\ast})\xrightarrow{L}N(0,I), 
	\end{align} 
	where  
	\begin{align*}
		V_{OPL}=(\frac{M_A}{p})^{-1}\frac{V_{cOPL}}{r}(\frac{M_A}{p})^{-1} 
	\end{align*}
	with 
	\begin{align*}
		V_{cOPL}=\sum_{i=1}^{p}\frac{A_iA^T_i\hat{z}^{\ast}\hat{z}^{\ast T}A_iA^T_i}{p^2\pi^{OPL}_i}=\sum_{i=1}^{p}	\mid \hat{\beta}_{rls(i)}	\mid \|A_i\|_2\sum_{i=1}^{p}\frac{A_iA^T_i\hat{z}^{\ast}\hat{z}^{\ast T}A_iA^T_i}{p^2	\mid \hat{\beta}_{rls(i)}	\mid \|A_i\|_2}.
	\end{align*}
\end{theorem}

Now,  we provide an estimation error bound of our algorithm.

\begin{theorem}\label{thm.6}
	To the assumptions of   Theorem  \ref{thm.3}, 
	add that 
	\begin{align}\label{rem4.3.0.1}
			s_3\|V^i\|_2\|{{y}}\|_2\leq	\mid \widetilde{\beta}_{(i)}	\mid \leq s_4\|V^i\|_2\|{{y}}\|_2,  \quad i=1,\cdots,p,
	\end{align}
	where $ \widetilde{\beta}_{(i)}$ is as in Algorithm \ref{IDSP} and $0< s_3\leq s_4$,   the initial value $ \hat{z}_0 $ is set as $0$, and let $r\geq   \frac{32{s_4c_2\rho}}{3s_3c_1{\epsilon}^2}\mathrm{ln}(\frac{4\rho}{\delta})$ with   $\epsilon,\delta \in (0,1)$ 	and  $m<\frac{1}{\delta}$. Then, for $  \widetilde{S} $ constructed by $\pi^{NOPL}_i$ and any $\epsilon$, 	with the probability at least $1-m\delta $,
	$\hat{\beta}_m$ generated from Algorithm \ref{IDSP} satisfies 
	\begin{align}\label{thm6.1}
		\|\hat{\beta}_m-\hat{\beta}_{rls}\|_2\leq\epsilon^m\|\hat{\beta}_{rls}\|_2.
	\end{align}

\end{theorem}

\begin{remark}\label{rem.thm6.1}
	The bound \eqref{thm6.1} can be used to determine the iteration number. Specifically,
	it is enough to do $log_{\epsilon}{\iota}$ iterations to get  $\|\hat{\beta}_m-\hat{\beta}_{rls}\|_2\leq\iota\|\hat{\beta}_{rls}\|_2$.
\end{remark}

\section{Numerical Experiments}
\label{sec.5}

In this section, we provide the numerical results of experiments with simulation data and real data. All experiments are implemented on a laptop  running MATLAB software with 16 GB random-access memory(RAM).

\subsection{Simulation Data--Example 1} \label{Sec.5.1.1}

In this example, the simulation data  is generated as done in \cite{slawski2018principal}. Specifically, we first produce an $n$-by-$p$ matrix  $B$   randomly, whose entries are drawn i.i.d. from the $N(0,1)$ distribution and SVD  is denoted as $U_B\Sigma_BV^T_B$ with $U_B\in \mathbb{R}^{n\times n}$, $\Sigma_B\in \mathbb{R}^{n\times n}$ and  $V_B\in \mathbb{R}^{p\times n}$. Then, we get $ A$ by  replacing $\Sigma_B$ with $\Sigma_0$, i.e, $ A=U_B\Sigma_0V^T_B $,  where  $\Sigma_0$
is a diagonal matrix with polynomial decay diagonal elements  $\sigma_j (j=1,\cdots,n)$, namely, 	
$\sigma_j\varpropto	9 \times j^{-8}$.
Furthermore,  we construct the response vector $y$ by 
$	y=A\beta+\varsigma$,
where $\beta \in \mathbb{R}^{p}$ and  $\varsigma \in \mathbb{R}^{n}$ have i.i.d. $N(0,1)$ entries.

In the specific experiments, we set $n=500$ and  $p=20000$. The description on parameters of the experiments is summarized in Table \ref{Tab5.1}, the explanation on six sampling methods is given in Table \ref{Tab5.2}, and the numerical results on  accuracy, i.e., the estimation error $\frac{\|\hat{\beta}_m-\hat{\beta}_{rls}\|_2}{\|\hat{\beta}_{rls}\|_2}$ and the prediction error $\frac{\|A\hat{\beta}_m-A\hat{\beta}_{rls}\|_2}{\|A\hat{\beta}_{rls}\|_2}$, 
and CPU time are  shown in Figures 1-4. Note that all the error results are on log-scale,  all the numerical results  are based on 50 replications of  Algorithm 
\ref{IDSP}, and it suffices to run the  step 2 in Algorithm \ref{IDSP} if $\pi^{OPL}_{i}$, $\pi^{LEV}_{i}$, $\pi^{RLEV}_{i}$, $\pi^{UNI}_{i}$ and $\pi^{COL}_{i}$ are used to generate $S$.   In addition, when $\pi^{OPL}_{i}$ is  employed, $C$ in Algorithm \ref{IDSP} should be   $(AA^T+\lambda I)^{-1}$, and when $\pi^{LEV}_{i}$, $\pi^{RLEV}_{i}$, $\pi^{UNI}_{i}$ and $\pi^{COL}_{i}$ are  adopted,  the lines 5--7  of the step 2  of Algorithm \ref{IDSP}  can  be omitted. 

\begin{table}[h]
	
		\scriptsize		
	\centering	
			\caption{Description of two  experiments for example 1.}%
\begin{threeparttable}
	
	\begin{tabular}{cccccccccc}
		
		\hline\noalign{\smallskip} 

					Kinds&	Comparison &  $r$ &  $\lambda$  &  $m$  & $r_0$  & Results  \\
	\noalign{\smallskip}\hline\noalign{\smallskip}
	\multirow{3}{*}{1 }&	 \multirow{3}{*}{six methods}				
   &  $500$ to $5000$ & $10$ &$3$ & 	$ 100 $ (NOPL)	&   Figs. 1-3(a)		
					\\		
					~&	~ &  $1000$  & $1$ to $50$ &$3$ & 	$ 100 $ (NOPL)	&   Figs. 1-3(b)			
					\\	
					~	&	~ & $1000$  & $10$ &$1$ to $15$	& 	$ 100 $ (NOPL)&   Figs. 1-3(c)
					\\	
					\hline\noalign{\smallskip}	
					2 &	 OPL and NOPL  &  $2000$ & $10$ &$3$ & $ 100 $ to $2000$ (NOPL)	&   Fig. 4		
					\\	
						\noalign{\smallskip}\hline		
	\end{tabular}

		\end{threeparttable}
	\label{Tab5.1}
\end{table}

\begin{table}[h]
	\scriptsize		
\centering	
			\caption{Explanation  of sampling methods with different probabilities.}%
\begin{threeparttable}
	
	\begin{tabular}{cccccccccc}
		
		\hline\noalign{\smallskip} 
		
					Method&  $\pi_{i}$&  Expression   \\
					\noalign{\smallskip}\hline\noalign{\smallskip}
					OPL &  	  $ \pi^{OPL}_i $& 	    ${	\mid \hat{\beta}_{rls(i)}	\mid \|A_i\|_2}/{\sum_{i=1}^{p}	\mid \hat{\beta}_{rls(i)}	\mid \|A_i\|_2}$	
					\\		
					NOPL &     	$\pi^{NOPL}_i$&      	$ {	\mid \widetilde{\beta}_{(i)}	\mid \|A_i\|_2}/{\sum_{i=1}^{p}	\mid \widetilde{\beta}_{(i)}	\mid \|A_i\|_2} $	
					\\		
					LEV &     	$\pi^{LEV}_i$ &     	${\|V^i\|^2_2}/{\sum_{i=1}^{p}\|V^i\|^2_2}$	
					\\		
					RLEV & 	    	$\pi^{RLEV}_i$&      	${\|X^i\|^2_2}/{\sum_{i=1}^{p}\|X^i\|^2_2}$	
					\\	
					COL & 	$\pi^{COL}_i$& 	  	${\|A_i\|^2_2}/{\sum_{i=1}^{p}\|A_i\|^2_2}$	
					\\	
					UNI &     $\pi^{UNI}_i$&  	     ${1}/{p}$	
					\\		
						\noalign{\smallskip}\hline
\end{tabular}

\end{threeparttable}
\label{Tab5.2}
\end{table}

In the first  experiment, we aim to show that the estimators established by OPL and NOPL have better performance. The corresponding numerical results are presented in Figures 1-3. 
From these figures, it is obvious to find that OPL and NOPL outperform other methods on estimation and prediction accuracy no matter what $r$, $\lambda$ and $m$ are, but they need more  computing time than COL and UNI. However, the improvement in accuracy is more than the sacrifice of calculation cost, and fortunately, OPL and NOPL are   cheaper than LEV and RLEV. What is more, we can observe that NOPL has extremely similar accuracy to OPL, and the former consumes less running time.  In addition, in most cases,  the  errors of all the methods decrease when  $r$, $\lambda$ and  $m$ increase.

\begin{figure}[H]
	\centering
	\includegraphics[width=0.9\textwidth]{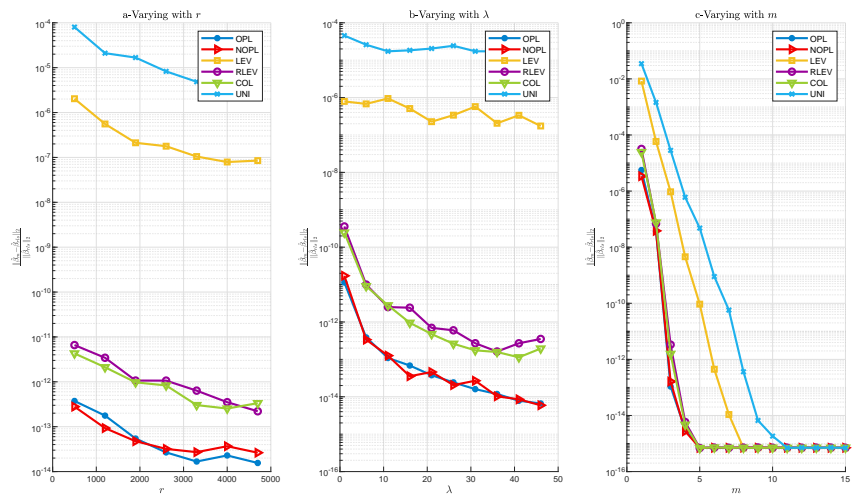}
	\caption{Comparison of  estimation errors using different methods for example 1.} 
	\label{fig5.1}
\end{figure}

\begin{figure}[H]
	\centering
	\includegraphics[width=0.9\textwidth]{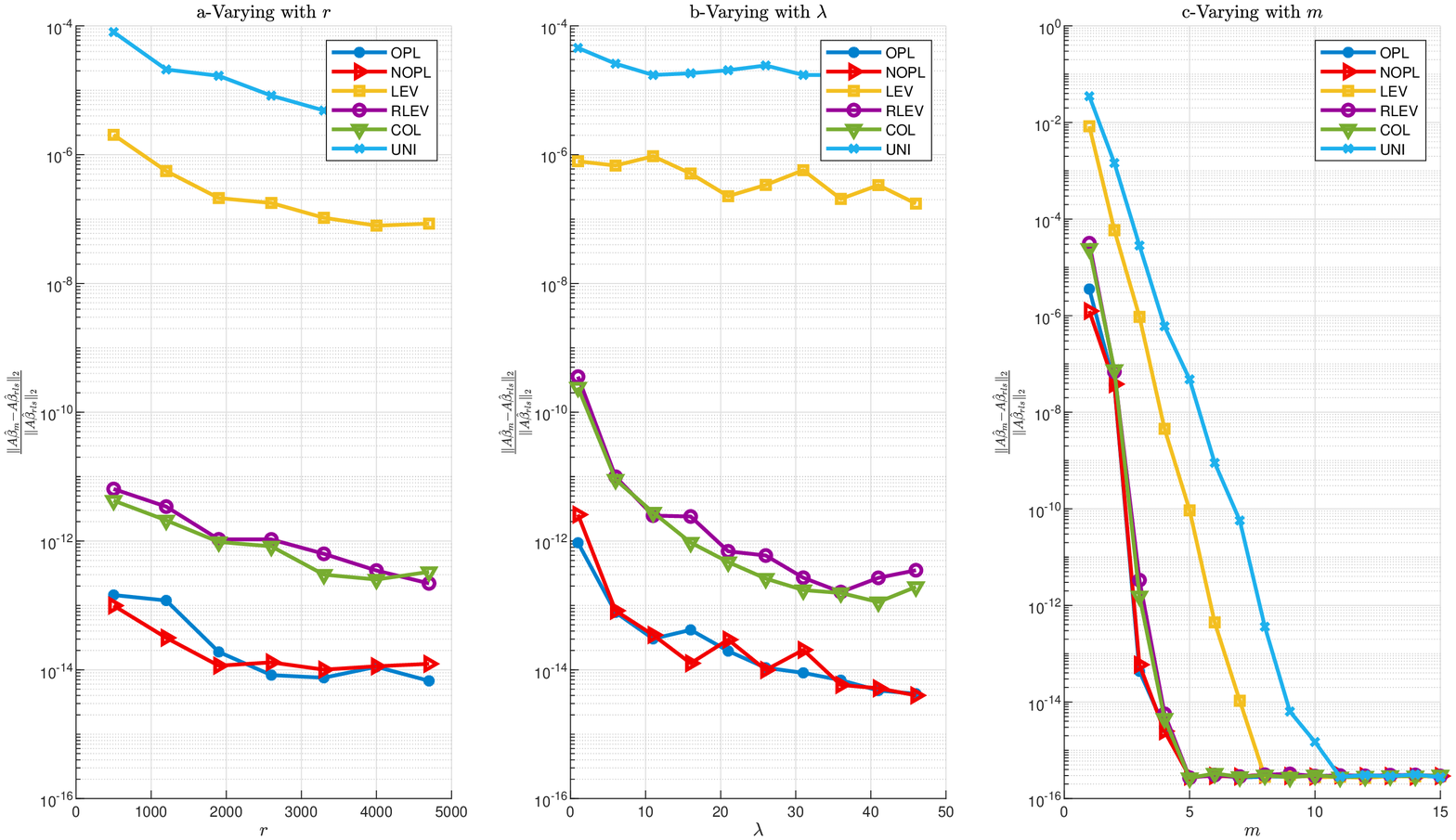}
	\caption{Comparison of  prediction errors using different methods for example 1.} 
	\label{fig5.2}
\end{figure}
\begin{figure}[H]
	\centering
	\includegraphics[width=0.9\textwidth]{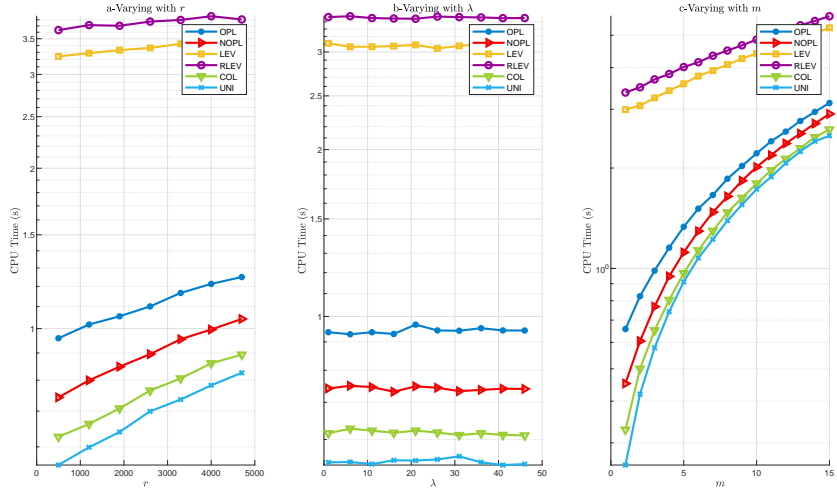} 
	\caption{Comparison of CPU time using different methods for example 1.} 
	\label{fig5.3}
\end{figure}

For the second  experiment, we compare the methods OPL and NOPL with different $r_0$. According to the  numerical results displayed in Figure 4,  it is evident to conclude that for different $r_0$, NOPL is able to  achieve significantly similar accuracy to OPL but spends less  computational cost. 

\begin{figure}[H]
	\centering
	\includegraphics[width=0.9\textwidth]{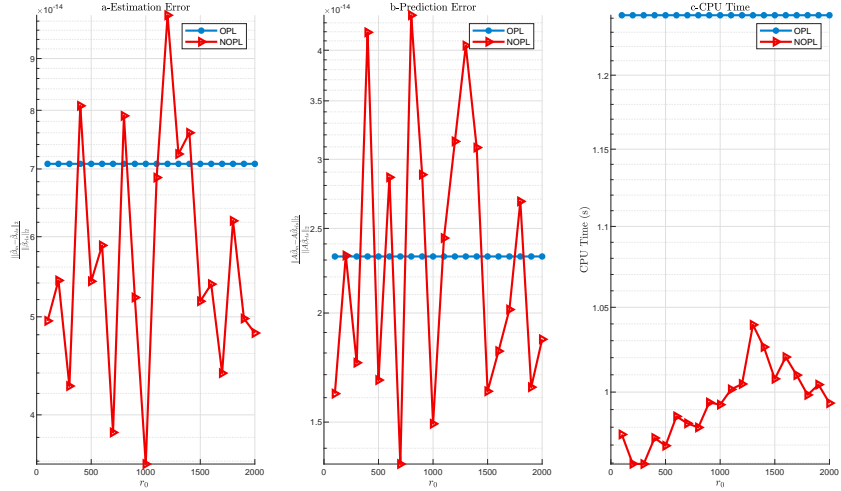}
	\caption{Comparison of OPL and NOPL with different $r_0$ for example 1.} 
	\label{fig5.4}
\end{figure}

\subsection{Simulation Data--Example 2} \label{Sec.5.1.2}

For this example, we produce the simulation data as done in \cite{chowdhury2018iterative}. Specifically, we construct an $n$-by-$p$ design  matrix 
$A=PDQ^T + \alpha M$,
where   $P \in \mathbb{R}^{n\times n}$ is a random matrix with i.i.d. $N(0,1)$ entries,   $D \in \mathbb{R}^{n\times n}$ is a  diagonal matrix with  diagonal entries $D_{ii}=(1-\frac{i-1}{p})^i ( i=1,\cdots n)$,  $Q \in \mathbb{R}^{p\times n}$ is a random column orthonormal matrix, 
$M\in \mathbb{R}^{n\times p}$ is a noise matrix with i.i.d. $N(0,1)$ entries, and $\alpha>0$ is a parameter used to  balance 
$PDQ^T$ and
$M$. 
In addition, the response vector $y\in \mathbb{R}^{n}$ is generated according to $y=A\beta + \gamma\varsigma$, where $\beta \in \mathbb{R}^{p}$ and  $\varsigma \in \mathbb{R}^{p}$ are constructed by i.i.d. $N(0,1)$ entries. In the specific experiments, we set $n=500$, $p=20000$, $\alpha=0.0001$ and $\gamma=0.5$, and repeat the  implementations in Section \ref{Sec.5.1.1} 
with different $r$, $r_0$, $\lambda$ and $m$ shown in Table \ref{Tab5.5}.

\begin{table}[h]
	\scriptsize		
	\centering
			\caption{Description of two experiments  for example 2.}%
		
				\begin{threeparttable}
				
				\begin{tabular}{cccccccccc}
					
					\hline\noalign{\smallskip} 
					
					Kinds&	Comparison&  $r$ &  $\lambda$  &  $m$  & $r_0$  & Results  \\
					
					\noalign{\smallskip}\hline\noalign{\smallskip}
				\multirow{3}{*}{1 }&	 \multirow{3}{*}{six methods} 
				 & $3000$ to $10000$ & $20$ &$15$ & $ 2000 $ (NOPL)	&  Figs. 5-7(a)		
					\\	
					~	& ~&  $5000$  & $1$ to $200$ &$15$ & $ 2000 $ (NOPL)&  Figs. 5-7(b)			
					\\	
					~& ~& $5000$  & $20$ &$1$ to $30$	& $ 2000 $ (NOPL) &  Figs. 5-7(c)
					\\	
						\hline\noalign{\smallskip}	
					
					2 &	OPL and NOPL   &  $5000$ & $20$ &$15$ & $ 500 $ to $20000$ (NOPL)&   Fig. 8		
					\\

					\noalign{\smallskip}\hline

				\end{tabular}
		\end{threeparttable}
	
	\label{Tab5.5}
\end{table}

From the  numerical results  presented in Figures 5-8, we can gain the similar observations to the ones in Section \ref{Sec.5.1.1}. 
That is,  taking different $r$,  $\lambda$ and $m$,
OPL and NOPL always  perform better  than other methods on accuracy, however, need more CPU time compared with COL and UNI. 
And, OPL and NOPL still show better computational efficiency than LEV and RLEV.  Besides, when setting a proper $r_0$ or a large  $\lambda$,  NOPL  
and OPL have similar accuracy but the former needs less running time. Unfortunately,  when $r_0$ is very large, NOPL  loses its advantage in CPU time. This is because in this case the computational cost of  $\widetilde{\beta}$ may not be less than that of $\hat{\beta}_{rls}$.

\begin{figure}[H]
	\centering
	\includegraphics[width=0.9\textwidth]{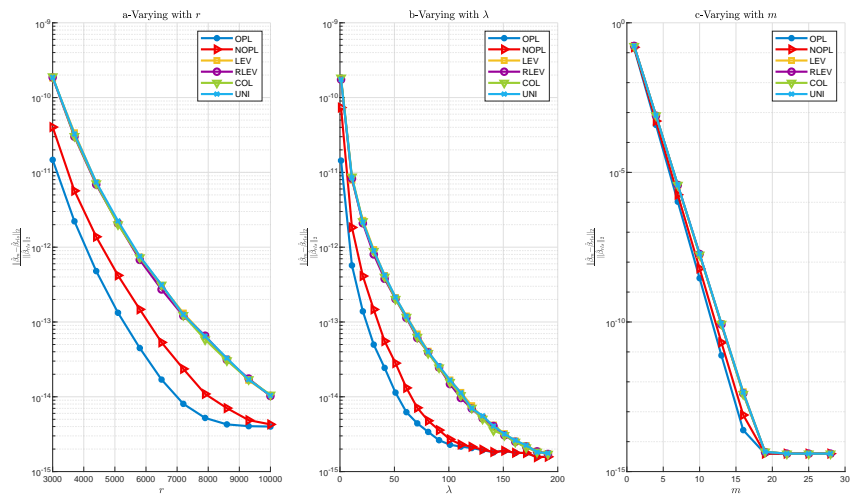} 
	\caption{Comparison of  estimation errors using different methods for example 2.} 
	\label{fig5.5}
\end{figure}

\begin{figure}[H]
	\centering
	\includegraphics[width=0.9\textwidth]{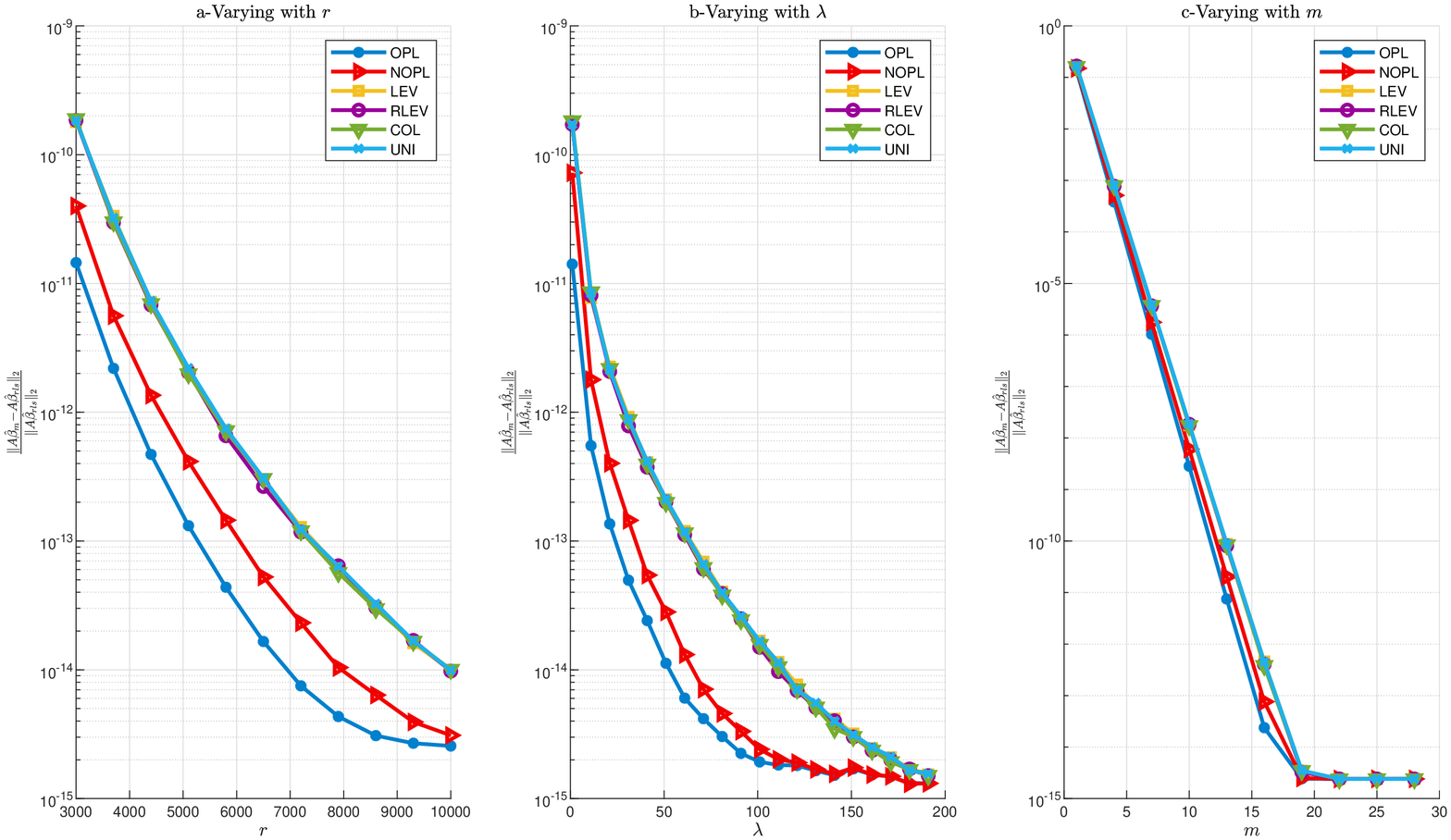} 
	\caption{Comparison of  prediction errors using different methods for example 2.} 
	\label{fig5.6}
\end{figure}
\begin{figure}[H]
	\centering
	\includegraphics[width=0.9\textwidth]{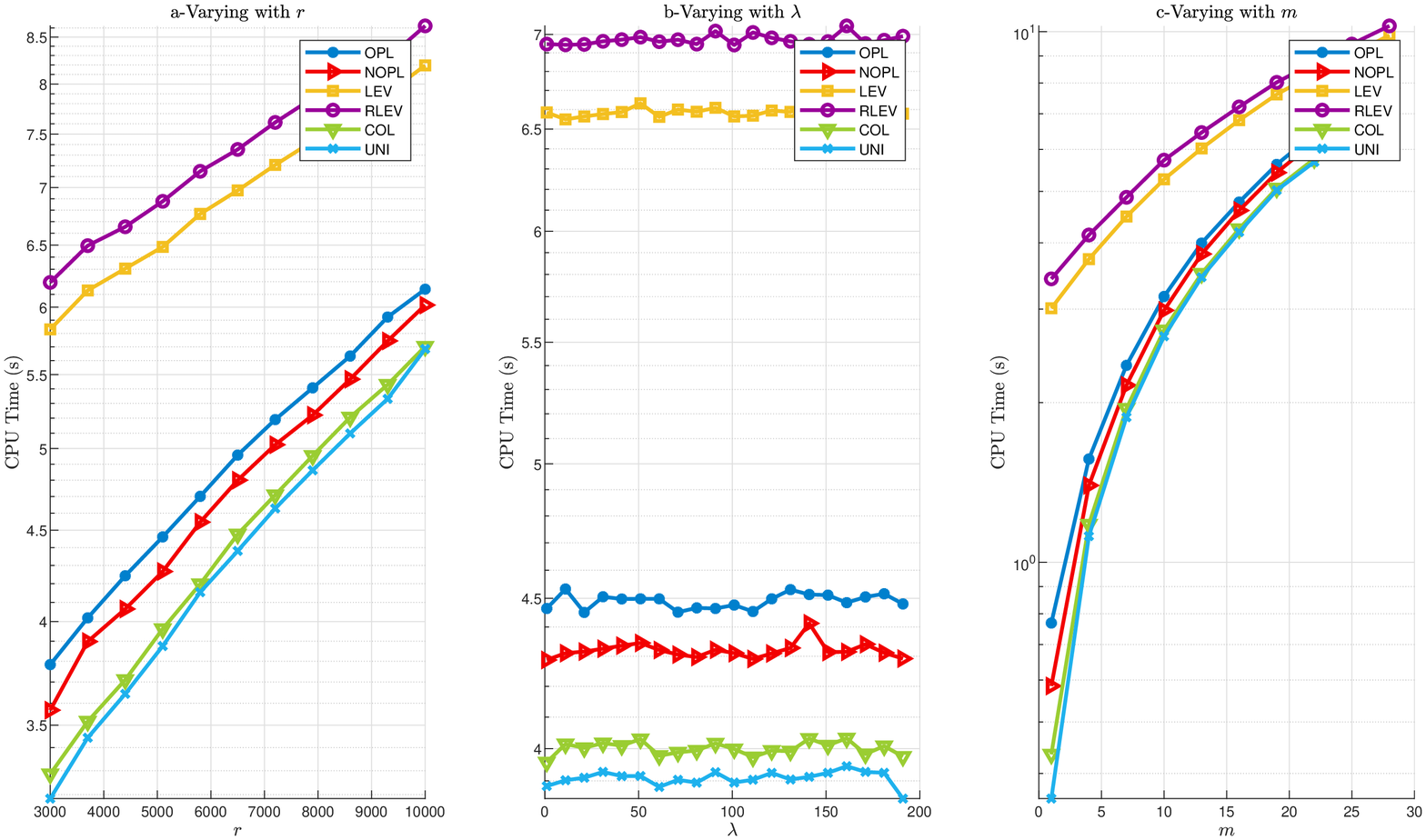}
	\caption{Comparison of CPU Time using different methods for example 2.} 
	\label{fig5.7}
\end{figure}

\begin{figure}[H]
	\centering
	\includegraphics[width=0.9\textwidth]{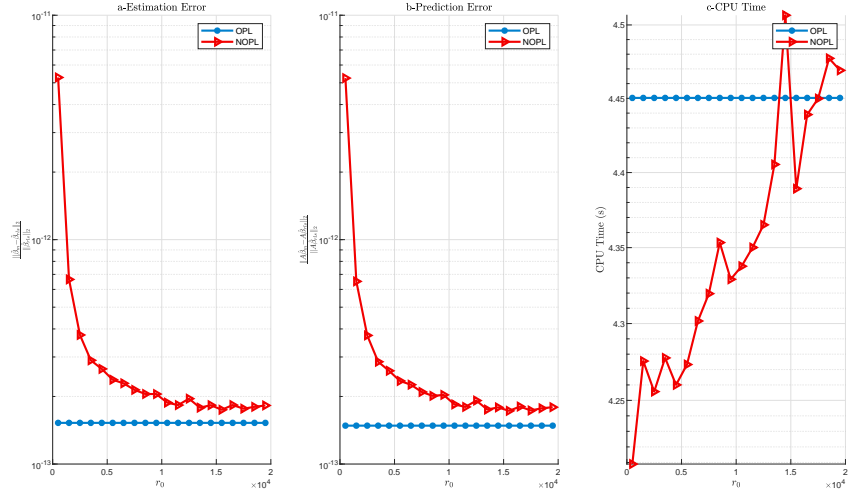} 
	\caption{Comparison of OPL and NOPL with different $r_0$ for example 2.} 
	\label{fig5.8}
\end{figure}


\subsection{Real Data--Gene Expression Cancer RNA-Seq Data Set}\label{Sec.5.2.1}

The data set is from the UCI machine learning repository, 
which can be found in \url{http://archive.ics.uci.edu/ml/datasets/gene+expression+cancer+RNA-Seq}.
Here, we only take the first 400 samples with 20531 real-valued features, and  centralize the design matrix. The response vector consists of 1, 2, 3, 4 and 5 labels, which represent five different types of tumors, i.e., PRAD, LUAD, BRCA, KIRC and COAD. We also centralize it. 

We  repeat  the experiments  in Sections \ref{Sec.5.1.1} and \ref{Sec.5.1.2} with different $r$, $r_0$, $\lambda$ and $m$. 
More details  are put in Table \ref{Tab5.3}.

\begin{table}[h]
	\scriptsize		
\centering
			\caption{Description of two experiments using Gene Expression Cancer RNA-Seq data set.}%
			\begin{threeparttable}
			
			\begin{tabular}{cccccccccc}
				
				\hline\noalign{\smallskip} 
					Kinds&	Comparison&  $r$ &  $\lambda$  &  $m$  & $r_0$  & Results  \\
					
						\noalign{\smallskip}\hline\noalign{\smallskip}
					\multirow{3}{*}{1 }&	 \multirow{3}{*}{six methods}   & $5000$ to $10000$ & $10$ &$16$ & $ 5000 $ (NOPL)	&  Figs. 9-11(a)		
					\\		
					~	& ~&  $8000$  & $1$ to $50$ &$16$ & $ 5000 $ (NOPL)&  Figs. 9-11(b)			
					\\	
					~& ~& $8000$  & $10$ &$1$ to $26$	& $ 5000 $ (NOPL) & Figs. 9-11(c)
					\\	
					
				\hline\noalign{\smallskip}	
					2&	OPL and NOPL   &  $8000$ & $10$ &$16$ & $ 1000 $ to $20531$ (NOPL)&   Fig. 12		
					\\	
				\noalign{\smallskip}\hline

		\end{tabular}
	
\end{threeparttable}

		\label{Tab5.3}
\end{table}

The  numerical results are displayed in Figures 9-12, and the conclusions summarized from these figures  are akin to the ones found in Sections \ref{Sec.5.1.1} and \ref{Sec.5.1.2}. 
Namely, compared with UNI and COL,  the accuracy of OPL and NOPL is dramatic improved at the cost of slightly  computational efficiency, and OPL performs better than LEV and RLEV on accuracy and  computing time. Although NOPL is only a little better than LEV and RLEV on accuracy,  it owns greatly advantage of CPU time. When taking a proper $r_0$, NOPL can be a well approximation of OPL but consumes less computing time. 
However, when $r_0$ is very large, NOPL will lose its superiority in computational cost.      In addition, for this real data, the choice of $\lambda$  has little influence on accuracy.

\begin{figure}[H]
	\centering
	\includegraphics [width=0.9\textwidth]{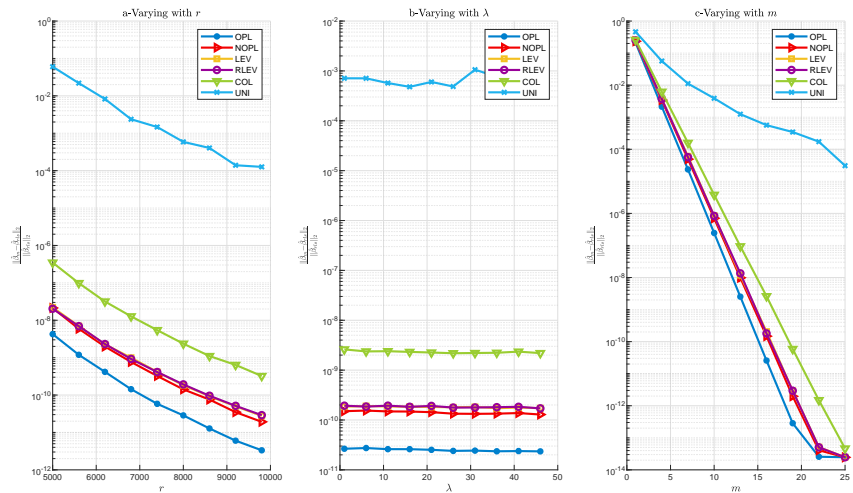}  
	\caption{Comparison of  estimation errors for different methods for  Gene Expression Cancer RNA-Seq data set.} 
	\label{fig5.9}
\end{figure}

\begin{figure}[H]
	\centering
	\includegraphics[width=0.9\textwidth]{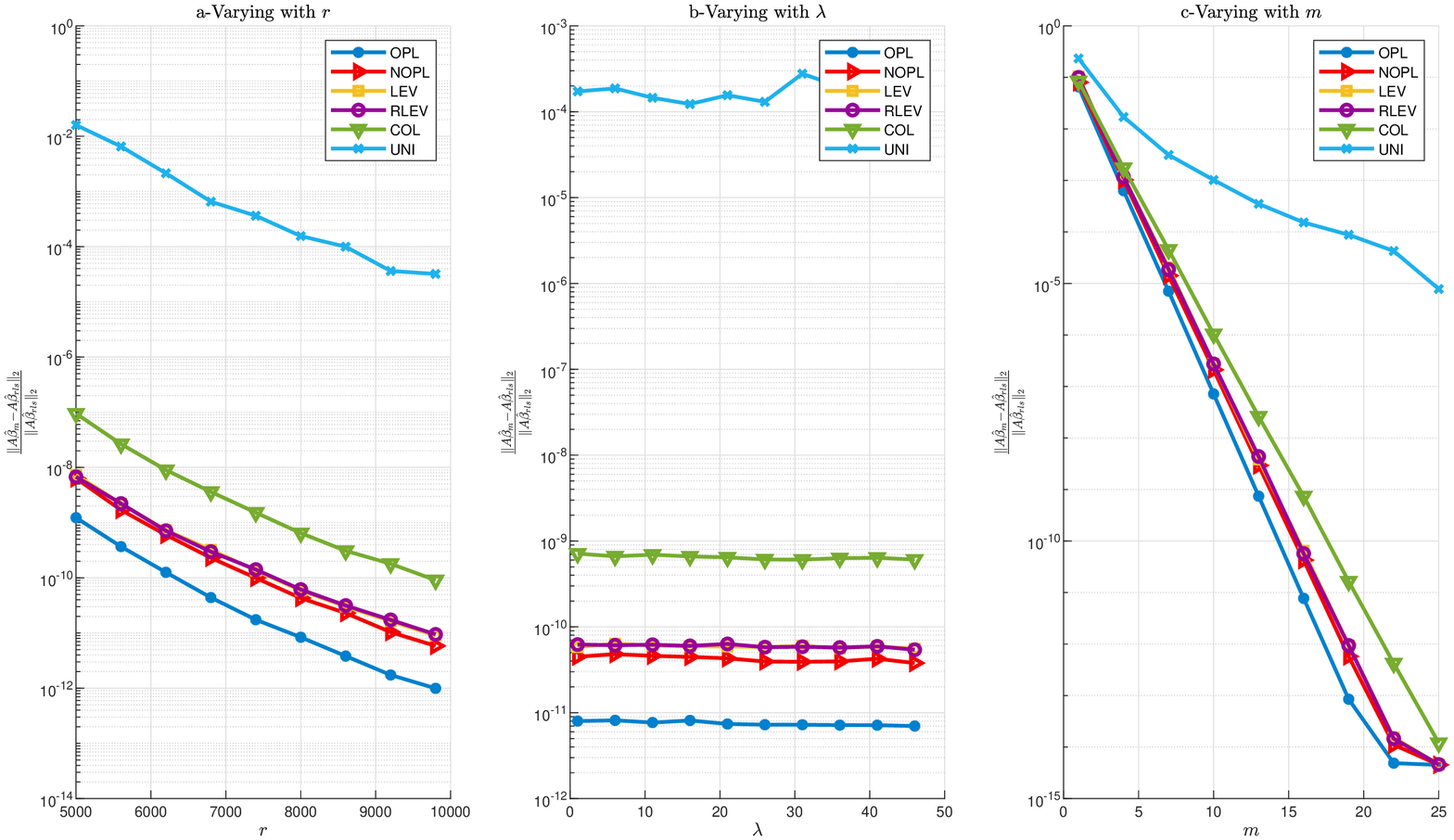} 
	\caption{Comparison of  prediction errors  for different methods for  Gene Expression Cancer RNA-Seq data set.} 
	\label{fig5.10}
\end{figure}
\begin{figure}[H]
	\centering
	\includegraphics[width=0.9\textwidth]{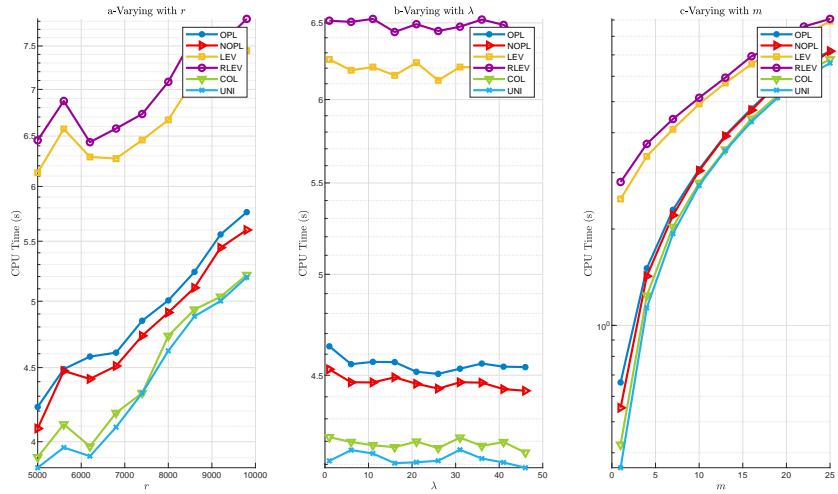}
	\caption{Comparison of CPU Time for different methods for  Gene Expression Cancer RNA-Seq data set.} 
	\label{fig5.11}
\end{figure}

\begin{figure}[H]
	\centering
	\includegraphics [width=0.9\textwidth]{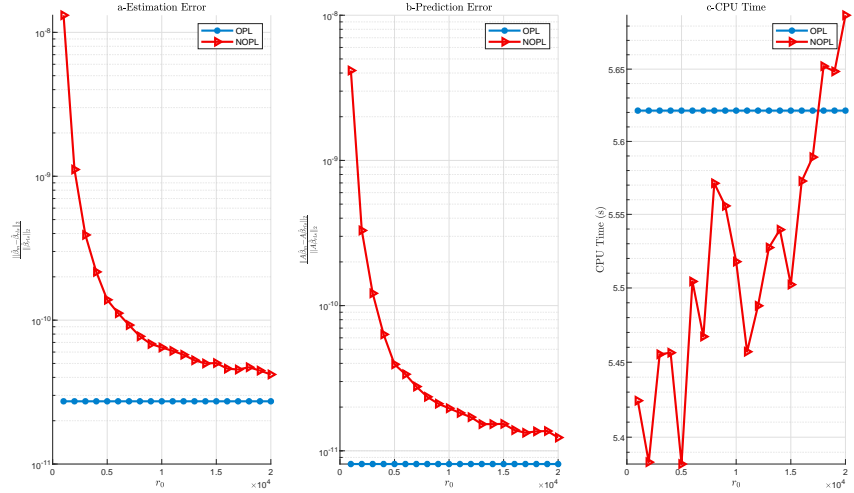} 
	\caption{Comparison of OPL and NOPL with different $r_0$ and Gene Expression Cancer RNA-Seq data set.} 
	\label{fig5.12}
\end{figure}

\subsection{Real Data--Gisette Data Set}\label{Sec.5.2.2}

This data set is also  from the UCI machine learning repository, 
which can be found in \url{http://archive.ics.uci.edu/ml/datasets/Gisette}.
In our experiments, the first 100 samples of training set with 5000 real-valued features are taken, and  the response vector is made up with $\pm1$ labels. Also, we centralize the response vector and design matrix prior to analysis. 

As done  in  Section  \ref{Sec.5.2.1},  
we also repeat  the experiments  in Sections \ref{Sec.5.1.1} and \ref{Sec.5.1.2} 
with different $r$, $r_0$, $\lambda$ and $m$. The detailed description can be found in Table \ref{Tab5.4}.

\begin{table}[h]
	\scriptsize		
\centering
			\caption{Description of two experiments using Gisette data set.}%
		\begin{threeparttable}
		
		\begin{tabular}{cccccccccc}
			
			\hline\noalign{\smallskip} 
			
					Kinds&	Comparison&  $r$ &  $\lambda$  &  $m$  & $r_0$  & Results  \\
					
					\noalign{\smallskip}\hline\noalign{\smallskip}
				\multirow{3}{*}{1 }&	 \multirow{3}{*}{six methods}  & $1000$ to $5000$ & $10$ &$10$ & $ 900 $ (NOPL)	&  Figs. 13-15(a)		
					\\		
					~	& ~&  $2000$  & $1$ to $50$ &$10$ & $ 900 $ (NOPL)&  Figs. 13-15(b)			
					\\	
					~	&~ & $2000$  & $10$ &$1$ to $26$	& $ 900 $ (NOPL) & Figs. 13-15(c)
					\\	
					
				\hline\noalign{\smallskip}
					2 &	OPL and NOPL   &  $2000$ & $10$ &$10$ & $ 500 $ to $5000$ (NOPL)&   Fig. 16		
					\\			
					
					\noalign{\smallskip}\hline

			\end{tabular}
			\end{threeparttable}
\label{Tab5.4}
\end{table}

The  numerical results are shown in Figures 13-16, and are almost identical with the observations in Section \ref{Sec.5.2.1}. 
To be more specific,   whatever the values of $r$,  $\lambda$ and $m$ are,  for accuracy, OPL and  NOPL always outperform other methods. Similarly, as for CPU time,   OPL and NOPL  are  slightly inferior to UNI and COL, but are greatly superior to LEV and RLEV. Only when $r_0$ is not particularly large,  NOPL  has good performance  on both accuracy and computing time, and  qualifies as a well alternative to OPL. Besides, the change of $\lambda$ also has little effect  on accuracy.

\begin{figure}[H]
	\centering
	\includegraphics[width=0.9\textwidth]{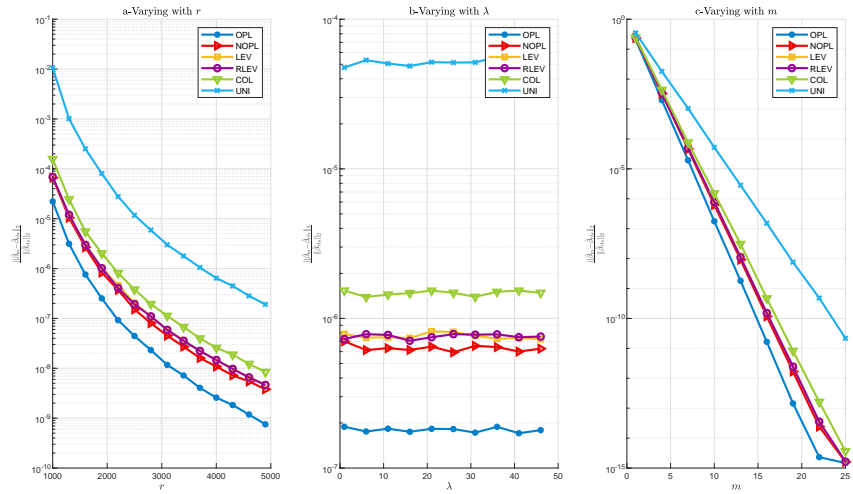}
	\caption{Comparison of  estimation errors for different methods for  Gisette data set.} 
	\label{fig5.13}
\end{figure}

\begin{figure}[H]
	\centering
	\includegraphics[width=0.9\textwidth]{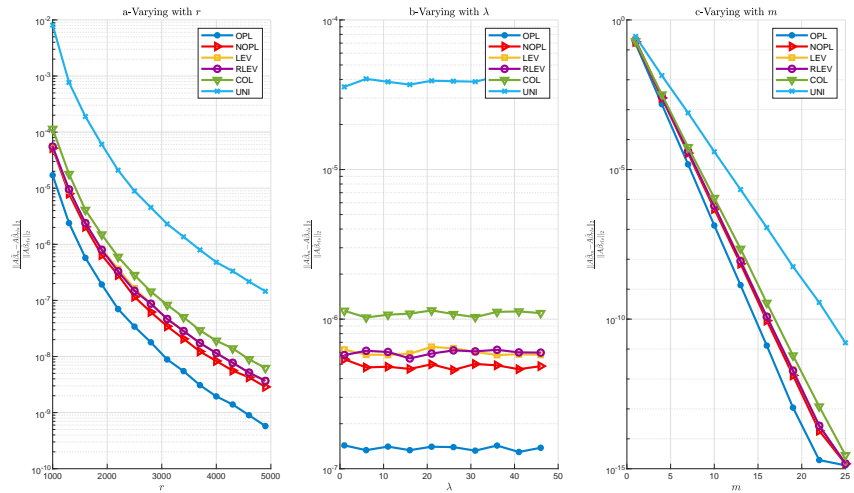} 
	\caption{Comparison of  prediction errors  for different methods for Gisette data set.} 
	\label{fig5.14}
\end{figure}
\begin{figure}[H]
	\centering
	\includegraphics[width=0.9\textwidth]{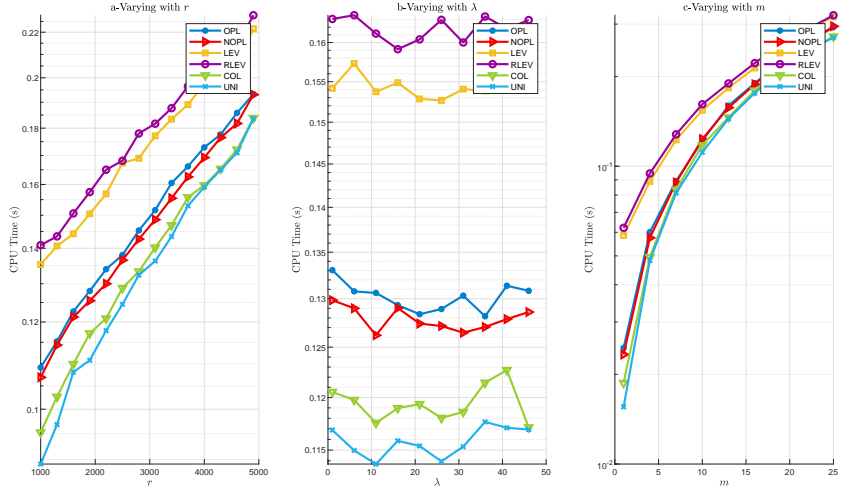} 
	\caption{Comparison of CPU Time for different methods for  Gisette data set.} 
	\label{fig5.15}
\end{figure}

\begin{figure}[H]
	\centering
	\includegraphics[width=0.9\textwidth]{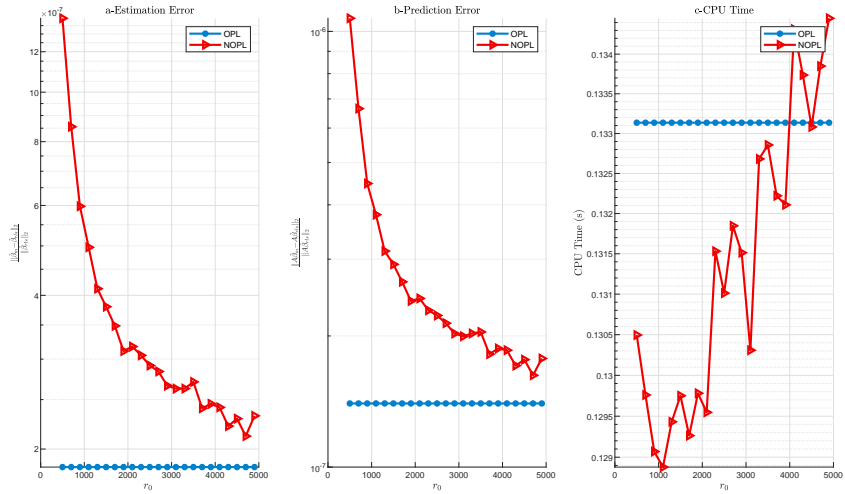} 
	\caption{Comparison of OPL and NOPL with different $r_0$ and Gisette data set.} 
	\label{fig5.16}
\end{figure}



\begin{appendices}
	
	

	
	\section{Proof of Theorem \ref{thm.1}}
	\label{app.A}

	We start by  establishing two lemmas.

	\begin{lemma}\label{lem.1}
		Assuming that  the conditions 
		\eqref{cond.3}, \eqref{cond.1} and  \eqref{rem.c2.2} are satisfied,   we have 	
		\begin{align}\label{lem1.1}
			\sum_{i=1}^{p}\pi_i\|\frac{e_i}{p}\|^{3}_2=O_p(1),
		\end{align}
		where $e_i=(\frac{A_iA^T_i}{\pi_i}+\lambda I)\hat{z}^{\ast}-\widetilde{y}$ with $\widetilde{y}=\lambda y$ and $\hat{z}^{\ast}$ being as in \eqref{rls.2}. 
	\end{lemma}
	
	\begin{proof}
		With $e_i=(\frac{A_iA^T_i}{\pi_i}+\lambda I)\hat{z}^{\ast}-\widetilde{y}$  and \eqref{rls.2}, 
		it is easy to see that  
		\begin{align*}
			\sum_{i=1}^{p}\pi_i\|\frac{e_i}{p}\|^{3}_2 =\frac{1}{p^{3}}\sum_{i=1}^{p}\pi_i\|(\frac{A_iA^T_i}{\pi_i}+\lambda I)(AA^T+\lambda I)^{-1}\widetilde{y}-\widetilde{y}\|_2^{3}.
		\end{align*}			
		Then, considering  the basic triangle inequality  
		and the fact that $\sum_{i=1}^{p}\pi_i=1$,  we can have 
		\begin{align}
			\sum_{i=1}^{p}\pi_i\|\frac{e_i}{p}\|^{3}_2
			\leq&\frac{1}{p^{3}}[\sum_{i=1}^{p}\pi_i\|(\frac{A_iA^T_i}{\pi_i}+\lambda I)(AA^T+\lambda I)^{-1}\widetilde{y}\|_2^{3}]+  \frac{\|\widetilde{y}\|_2^{3}}{p^{3}}
			\nonumber \\
			&+3\frac{1}{p^{3}}[\sum_{i=1}^{p}\pi_i\|(\frac{A_iA^T_i}{\pi_i}+\lambda I)
			(AA^T+\lambda I)^{-1}\widetilde{y}\|_2^{2}\|\widetilde{y}\|_2]
			\nonumber \\
			&+3\frac{1}{p^{3}}[\sum_{i=1}^{p}\pi_i\|(\frac{A_iA^T_i}{\pi_i}+\lambda I)(AA^T+\lambda I)^{-1}\widetilde{y}\|_2\|\widetilde{y}\|^2_2]\nonumber \\
			\leq&\frac{\|\widetilde{y}\|_2^{3}\sigma^{3}_{n}(AA^T+\lambda I)}{p^{3}}(\sum_{i=1}^{p}\frac{\|A_iA^T_i\|^3_2}{\pi^2_i}+3\lambda\sum_{i=1}^{p}\frac{\|A_iA^T_i\|^2_2}{\pi_i}
			\nonumber \\
			&+3\lambda^2\sum_{i=1}^{p}{\|A_iA^T_i\|_2}+\lambda^3)+\frac{\|\widetilde{y}\|_2^{3}}{p^{3}}
			\nonumber \\
			&+3\frac{\|\widetilde{y}\|_2^{3}\sigma^{2}_{n}(AA^T+\lambda I)}{p^{3}}(\sum_{i=1}^{p}\frac{\|A_iA^T_i\|^2_2}{\pi_i}+2\lambda\sum_{i=1}^{p}{\|A_iA^T_i\|_2}+\lambda^2)
			\nonumber \\ 
			&+3\frac{\|\widetilde{y}\|_2^{3}\sigma_{n}(AA^T+\lambda I)}{p^{3}}(\sum_{i=1}^{p}{\|A_iA^T_i\|_2}+\lambda).
			\label{proof.lem.1.0}		
		\end{align}
	Following
			\begin{align}\label{proof.lem.1.01}
				\frac{\|\widetilde{{{y}}}\|^2_2}{p}=o_p(1),
			\end{align}
			which can be  derived from $ np^{-1}\to 0 $, and noting \eqref{cond.3}, \eqref{cond.1}, 
			 \eqref{rem.c2.2} and  \eqref{proof.lem.1.0}, we can get
			\begin{align*}
				\sum_{i=1}^{p}\pi_i\|\frac{e_i}{p}\|^{3}_2\leq& \frac{\|\widetilde{y}\|_2^{3}\sigma^{3}_{n}(AA^T+\lambda I)}{p^{3}}(\sum_{i=1}^{p}\frac{\|A_i\|^6_2}{\pi^2_i})+o_p(1) 
				\quad \rm{by}  \ \eqref{cond.1}, 
				\  \eqref{rem.c2.2}, \ \eqref{proof.lem.1.0}, \ and \ \eqref{proof.lem.1.01} \nonumber \\
				=&{O}_p(1). \quad \rm{by} \ \eqref{cond.3}  \nonumber
		\end{align*}
		Thus, \eqref{lem1.1} is arrived. 
	\end{proof}

	\begin{lemma}\label{lem.2}
		
		Suppose that the conditions  \eqref{cond.1} and \eqref{rem.c2.1}  hold. 
		Then, conditional on  $\mathcal{F}_n$   in probability,	
		\begin{align}
			\frac{\widehat{M}_A-M_A}{p}=O_{p	\mid \mathcal{F}_n}(r^{-1/2}), \label{lem2.1}\\
			\mathit{\frac{e^{\ast}}{p}=O_{p	\mid \mathcal{F}_n}(r^{\rm{-1/2}})}, \label{lem2.2}
		\end{align}
		where $M_A=AA^T+\lambda I$, $\hat{M}_A=ASS^TA^T+\lambda I$ with $S\in \mathbb{R}^{p \times r}$ constructed as in Algorithm 1, 
		and $e^{\ast}=(\hat{M}_A\hat{z}^{\ast}-\widetilde{y})$ with $\widetilde{y}=\lambda y$ and  $ \hat{z}^{\ast} $ being as in \eqref{rls.2}.
	\end{lemma}
	
	\begin{proof}
		First, note that
		\begin{align*}%
			&\frac{1}{p^2}\sum_{i=1}^{p}\pi_i[(\frac{A_{i}A_{i}^T}{\pi_{i}}+\lambda I)-(AA^T+\lambda I)][(\frac{A_{i}A_{i}^T}{\pi_{i}}+\lambda I)-(AA^T+\lambda I)]\nonumber\\
			&=\frac{1}{p^2}\sum_{i=1}^{p}\pi_i(\frac{A_{i}A_{i}^T}{\pi_{i}}-AA^T)(\frac{A_{i}A_{i}^T}{\pi_{i}}-AA^T)\\
			& =\frac{1}{p^2}\sum_{i=1}^{p}\frac{A_{i}A_{i}^TA_{i}A_{i}^T}{\pi_{i}}
			-\frac{AA^TAA^T}{p^2} \nonumber \\
			&=O_p(1),  
		\end{align*}
		where the last equality is from \eqref{cond.1}  and \eqref{rem.c2.1}.
		This result implies, for any $n$-dimensional vector $\ell$ with finite elements, 
		\begin{align} \label{prooflem2.4}
			&\frac{1}{p^2}\sum_{i=1}^{p}\pi_i[(\frac{A_{i}A_{i}^T}{\pi_{i}}+\lambda I)-(AA^T+\lambda I)]\ell\ell^{T}[(\frac{A_{i}A_{i}^T}{\pi_{i}}+\lambda I)-(AA^T+\lambda I)]\nonumber\\			&=\frac{1}{p^2}\sum_{i=1}^{p}\frac{A_{i}A_{i}^T\ell\ell^{T}A_{i}A_{i}^T}{\pi_{i}}
			-\frac{AA^T\ell\ell^{T}AA^T}{p^2} =O_p(1).
		\end{align}
		Thus, following $\mathrm{E}(\hat{M}_A	\mid  A)=M_A$, it is natural to get
		\begin{align*}
			\mathrm{Var}(\frac{(\hat{M}_A-M_A)\ell}{p}	\mid  A)=&\mathrm{E}[(\frac{\hat{M}_A-M_A}{p})\ell\ell^T(\frac{\hat{M}_A-M_A}{p})	\mid  A] \\
			=&\frac{1}{rp^2}\sum_{i=1}^{p}\pi_i[(\frac{A_{i}A_{i}^T}{\pi_{i}}+\lambda I)-(AA^T+\lambda I)]\ell	\ell^{T}[(\frac{A_{i}A_{i}^T}{\pi_{i}}+\lambda I)-(AA^T+\lambda I)]\\
			=&O_p(r^{-1}),
		\end{align*}
		which together with the Markov's inequality implies	\eqref{lem2.1}.  
		
		Combining  \eqref{prooflem2.4} and \eqref{proof.lem.1.01}, 
		we can  get
		\begin{align}\label{prooflem2.7}
			&	\frac{1}{p^2}\sum_{i=1}^{p}\pi_i\hat{z}^{\ast T}(\frac{A_{i}A_{i}^T}{\pi_{i}}+\lambda I)\ell\ell^{T}(\frac{A_{i}A_{i}^T}{\pi_{i}}+\lambda I)\hat{z}^{\ast}  \nonumber\\
			&=\frac{1}{p^2}\hat{z}^{\ast T}(AA^T+\lambda I)\ell\ell^{T}(AA^T+\lambda I)\hat{z}^{\ast}+O_p(1).
		\end{align}	
		Thus,	considering 
		$e^{\ast}=\sum_{t=1}^{r}\frac{1}{r}e_{i_t}$ with  $e_{i_t}=(\frac{A_{i_t}A_{i_t}^T}{\pi_{i_t}}+\lambda I)\hat{z}^{\ast}-\widetilde{y}$ and $  \mathrm{E}(e_{i_t}	\mid \mathcal{F}_n)=0 $,  
		and \eqref{prooflem2.7}, we can obtain 
		\begin{align}
			\mathrm{Var}(\frac{\ell^T e^{\ast}}{p}	\mid \mathcal{F}_n)=&\ell^T\mathrm{E}[(\frac{e^{\ast}}{p})(\frac{e^{\ast}}{p})^T	\mid \mathcal{F}_n]\ell=\frac{1}{rp^2}\ell^T(\sum_{i=1}^{p}\pi_ie_ie^T_i)\ell\nonumber \\
			=&\frac{1}{rp^2}\ell^T[\sum_{i=1}^{p}\pi_i{((\frac{A_{i}A_{i}^T}{\pi_{i}}+\lambda I)\hat{z}^{\ast}-\widetilde{y})((\frac{A_{i}A_{i}^T}{\pi_{i}}+\lambda I)\hat{z}^{\ast}-\widetilde{y})^T}]\ell\nonumber \\
			=&\frac{1}{rp^2}\sum_{i=1}^{p}\pi_i\hat{z}^{\ast T}(\frac{A_{i}A_{i}^T}{\pi_{i}}+\lambda I)\ell\ell^{T}
			(\frac{A_{i}A_{i}^T}{\pi_{i}}+\lambda I)\hat{z}^{\ast}-\frac{\ell^T\widetilde{y}\widetilde{y}^T\ell}{rp^2}\nonumber \\
			=&\frac{1}{r}[\frac{1}{p^2}\hat{z}^{\ast T}(AA^T+\lambda I)\ell\ell^{T}(AA^T+\lambda I)\hat{z}^{\ast}+O_p(1)-\frac{\ell^T\widetilde{y}\widetilde{y}^T\ell}{p^2}] \quad \rm{by} \  \eqref{prooflem2.7} \nonumber \\
			=&\frac{1}{r}[\frac{\widetilde{y}^T\ell\ell^T\widetilde{y}}{p^2}+O_p(1)-\frac{\ell^T\widetilde{y}\widetilde{y}^T\ell}{p^2}] 
			\nonumber \\
			=&	O_p(r^{-1}).\nonumber 
		\end{align}
		Consequently, by the Markov's inequality, \eqref{lem2.2} is obtained.  
	\end{proof}	
	
	\noindent	${\textbf{Proof  of  Theorem \ref{thm.1}}}$.  
	Considering that  
	\begin{align*}
		\hat{z}&=(ASS^TA^T+\lambda I)^{-1}\widetilde{y}=\hat{M}^{-1}_A\widetilde{y}, \\ \hat{z}^{\ast}&=(AA^T+\lambda I)^{-1}\widetilde{y}=M^{-1}_A\widetilde{y},
	\end{align*}
	where  $\widetilde{y}=\lambda y$,
	we can  rewrite $	\hat{z}-\hat{z}^{\ast}$ as 
	\begin{align}
		\hat{z}-\hat{z}^{\ast}
		&=(ASS^TA^T+\lambda I)^{-1}(\widetilde{y}-(ASS^TA^T+\lambda I)\hat{z}^{\ast})\nonumber\\
		&=\hat{M}^{-1}_A(\widetilde{y}-\hat{M}_A\hat{z}^{\ast})=-\hat{M}^{-1}_Ae^{\ast}\nonumber\\
		&=-(\hat{M}^{-1}_A-M^{-1}_A+M^{-1}_A)e^{\ast}\nonumber\\
		&=-M^{-1}_Ae^{\ast}-(\hat{M}^{-1}_A-M^{-1}_A)e^{\ast}\nonumber
		\\
		&=-M^{-1}_Ae^{\ast}+\hat{M}^{-1}_A(\hat{M}_A-M_A)M^{-1}_Ae^{\ast}\nonumber\\
		&=-(\frac{M_A}{p})^{-1}\frac{e^{\ast}}{p}+(\frac{\hat{M}_A}{p})^{-1}(\frac{\hat{M}_A-M_A}{p})(\frac{M_A}{p})^{-1}\frac{e^{\ast}}{p}\label{proofthem.1.0.0}\\
		&=-(\frac{M_A}{p})^{-1}\frac{e^{\ast}}{p}+O_{p	\mid \mathcal{F}_n}(r^{-1}),\label{proofthem.1.0}
	\end{align}
	where the last equality is derived by 
	\eqref{rem.c2.1} and
	Lemma \ref{lem.2}.
	Thus, 	to prove \eqref{thm1.1},
	we first prove
	\begin{align}\label{proofthem.1.0.1}
		(\frac{V_c}{r})^{-1/2}(\frac{e^{\ast}}{p})\xrightarrow{L}N(0,I).
	\end{align}   
	
	Recall that $\frac{e_{\ast}}{p}=\sum_{t=1}^{r}\frac{1}{rp}e_{i_t}$ with 
	\begin{align*}
		e_{i_t}=(\frac{A_{i_t}A_{i_t}^T}{\pi_{i_t}}+\lambda I)\hat{z}^{\ast}-\widetilde{y}. 
	\end{align*}
	Now, we construct the sequence  $ 
	\{\frac{e_{i_t}}{p}\}^r_{t=1} $. These  random vectors  are independent and identically distributed and it  
	is easy to get 
	that  $ \mathrm{E}(\frac{e_{i_t}}{p}	\mid \mathcal{F}_n)=0 $.
	Furthermore, noting that
		\begin{align}\label{proofthm.1.2}
			V_c=	\sum_{i=1}^{p}\frac{A_{i}A_{i}^T\hat{z}^{\ast}\hat{z}^{\ast T}A_{i}A_{i}^T}{p^2\pi_{i}}=O_p(1),	
		\end{align}
		which can be obtained from \eqref{cond.1},  together with \eqref{rem.c2.1}  and \eqref{proof.lem.1.01},  
	we  have 
	\begin{align}
		\mathrm{Var}(\frac{e_{i_t}}{p}	\mid \mathcal{F}_n)=&\mathrm{E}(\frac{e_{i_t}e^T_{i_t}}{p^2}	\mid \mathcal{F}_n) 
		=\sum_{i=1}^{p}\pi_i\frac{e_{i}e^T_{i}}{p^2} \nonumber \\
		=&\sum_{i=1}^{p}\pi_i\frac{[(\frac{A_{i}A_{i}^T}{\pi_{i}}+\lambda I)\hat{z}^{\ast}-\widetilde{y}][(\frac{A_{i}A_{i}^T}{\pi_{i}}+\lambda I)\hat{z}^{\ast}-\widetilde{y}]^T}{p^2} \nonumber \\
		=&\sum_{i=1}^{p}\pi_ip^{-2}\frac{A_{i}A_{i}^T}{\pi_{i}}\hat{z}^{\ast}\hat{z}^{\ast T}\frac{A_{i}A_{i}^T}{\pi_{i}}+\frac{(\lambda\hat{z}^{\ast}-\widetilde{y})\hat{z}^{\ast T}AA^T}{p^2}
		\nonumber\\
		&+\frac{AA^T\hat{z}^{\ast}(\lambda\hat{z}^{\ast}-\widetilde{y})^T}{p^2} + \frac{(\lambda\hat{z}^{\ast}-\widetilde{y})(\lambda\hat{z}^{\ast}-\widetilde{y})^T}{p^2} \nonumber \\
		=&\sum_{i=1}^{p}\frac{A_{i}A_{i}^T\hat{z}^{\ast}\hat{z}^{\ast T}A_{i}A_{i}^T}{p^2\pi_{i}}+o_p(1) \quad  \rm{by} \ \eqref{rem.c2.1} \ and \ \eqref{proof.lem.1.01} \nonumber  \\
		=&V_c+o_p(1) \label{proofthm.1.2.0} \\
		=&O_p(1). \quad  \rm{by} \ \eqref{proofthm.1.2}  
		\label{proof1.2}
	\end{align}
	In addition, 	for any $\xi>0$,
	we have 
	\begin{align}
		&\sum_{t=1}^{r}\mathrm{E}[\|r^{-\frac{1}{2}}p^{-1}e_{i_t}\|^2_2I({\|r^{-\frac{1}{2}}p^{-1}e_{i_{t}}\|_2}>\xi)	\mid \mathcal{F}_n] \nonumber \\
		&=\sum_{i=1}^{p}\pi_i\|p^{-1}e_i\|^2_2I({\|r^{-\frac{1}{2}}p^{-1}e_{i}\|_2}>\xi) \nonumber \\
		&\leq (r^{\frac{1}{2}}\xi)^{-1}\sum_{i=1}^{p}\pi_i\|p^{-1}e_i\|^{3}_2   \nonumber \\
		&=o_p(1), \nonumber  
	\end{align}
	where  the  inequality  is deduced by the   constraint   $I({\|r^{-\frac{1}{2}}p^{-1}e_{i}\|_2}>\xi)$, and the last equality  is from  Lemma \ref{lem.1}. 
	Putting the above discussions together, we find that the Lindeberg-Feller conditions are satisfied in probability. Thus,   by  the  Lindeberg-Feller central limit theorem  \cite[Proposition 2.27]{van2000asymptotic}, and noting  \eqref{proof1.2}, we can acquire
	\begin{align} 
		[\mathrm{Var}(\frac{e_{i_t}}{p}	\mid \mathcal{F}_n)]^{-1/2}({r^{-1/2}p^{-1}}\sum_{t=1}^{r}e_{i_t})\xrightarrow{L}N(0,I),\nonumber
	\end{align}
	which combined with $\frac{e_{\ast}}{p}={r^{-1}p^{-1}}\sum_{t=1}^{r}{e_{i_t}}$ and  $\mathrm{Var}(\frac{e_{\ast}}{p}	\mid \mathcal{F}_n)=r^{-1}\mathrm{Var}(\frac{e_{i_t}}{p}	\mid \mathcal{F}_n)$  gives 
	\begin{align*}
		[r^{-1}\mathrm{Var}(\frac{e_{i_t}}{p}	\mid \mathcal{F}_n)]^{-1/2}(\frac{e_{\ast}}{p})\xrightarrow{L}N(0,I).
	\end{align*}	
	Thus, by Lemma \ref{lem.2},
	\eqref{proofthm.1.2.0},  and the  Slutsky's Theorem \cite[Theorem 6]{fergusoncourse},  we can get 	\eqref{proofthem.1.0.1}.	
	
	Now, we  prove \eqref{thm1.1}. 
	Following       \eqref{rem.c2.1} and \eqref{proofthm.1.2},	it is  easy to get
	\begin{align*}
		V=(\frac{M_A}{p})^{-1}\frac{V_c}{r}({\frac{M_A}{p}})^{-1}
		=O_p({r}^{-1}), 
	\end{align*}
	which together with \eqref{proofthem.1.0} yields
	\begin{align}
		V^{-1/2}(\hat{z}-\hat{z}^{\ast})&=-V^{-1/2}(\frac{M_A}{p})^{-1}\frac{e_{\ast}}{p}+O_{p	\mid \mathcal{F}_n}({r}^{-1/2})\nonumber\\
		&=-V^{-1/2}(\frac{M_A}{p})^{-1}(\frac{V_c}{r})^{1/2}(\frac{V_c}{r})^{-1/2}\frac{e_{\ast}}{p}+O_{p	\mid \mathcal{F}_n}({r}^{-1/2}). \label{proofthm1.6}
	\end{align}
	In addition, it is verified that
	\begin{align}\label{proofthm1.7}
		&V^{-1/2}(\frac{M_A}{p})^{-1}(\frac{V_c}{r})^{1/2}[V^{-1/2}(\frac{M_A}{p})^{-1}(\frac{V_c}{r})^{1/2}]^T \nonumber \\
		&=V^{-1/2}(\frac{M_A}{p})^{-1}(\frac{V_c}{r})^{1/2}(\frac{V_c}{r})^{1/2}(\frac{M_A}{p})^{-1}V^{-1/2}=I.
	\end{align}
	Thus, combining \eqref{proofthem.1.0.1}, \eqref{proofthm1.6},  and \eqref{proofthm1.7},  by the Slutsky's Theorem, we get the desired result  \eqref{thm1.1}. 



	\section{Proof of Theorem \ref{thm.2}} 
	\label{supp.thm2}
	
	
	According to the Cauchy-Schwarz inequality, we have 
	\begin{align*}
		\rm{tr}(\mathit{V}_\mathit{c})&=\sum_{i=1}^{p}\frac{A^T_i\hat{z}^{\ast}\hat{z}^{\ast T}A_i\|A_i\|^2_2}{p^2\pi_i}=\lambda^2\sum_{i=1}^{p}\frac{\hat{\beta}^2_{rls(i)}\|A_i\|^2_2}{p^2\pi_i}\nonumber \\
		&=\frac{\lambda^2}{p^2}\sum_{i=1}^{p}\pi_i\sum_{i=1}^{p}\frac{\hat{\beta}^2_{rls(i)}\|A_i\|^2_2}{\pi_i}\nonumber \geq\frac{\lambda^2}{p^2}(\sum_{i=1}^{p}|\hat{\beta}_{rls(i)}|\|A_i\|_2)^2,
	\end{align*}
	where the equality in the last inequality holds if and only if $\pi_i$ is proportional to $	\mid \hat{\beta}_{rls(i)}	\mid \|A_i\|_2$ for some  constant $C_0\geq0$. Thus, following $\sum_{i=1}^{p}\pi_i=1$,  the desired result  \eqref{thm2.1}
	is obtained. 

	
\section{Proof of Theorem \ref{thm.3}}
	\label{app.C}

	
	We first present two  auxiliary lemmas.
	
	\begin{lemma} \label{lem.4} \cite[Theorem 23]{chowdhury2018iterative}
		If $J $, $H $ $\in \mathbb{R}^{m\times m}$  are real symmetric positive semi-definite matrices such that  $\sigma_1(J)\geq\sigma_2(J)\geq \cdots \geq\sigma_m(J)$ and	 $\sigma_1(H)\geq\sigma_2(H)\geq \cdots \geq\sigma_m(H)$, then  
		\begin{align*}
			\sigma_j(J-H)\leq 	\sigma_j \begin{pmatrix}
				J & 0 \\ 0 &  H
			\end{pmatrix},  \quad {j=1,\cdots, m}.
		\end{align*}
		Especially, 
		\begin{align*}
			\|J-H\|_2\leq\mathop{\rm{max} }\{\|J\|_2,\|H\|_2\}.
		\end{align*}
	\end{lemma}
	
	\begin{lemma}\label{lem.5}
		For $S$ established by 
		$\pi_i=\pi^{OPL}_{i}$,
		assuming that  \eqref{lem.5.1} holds
		and  letting $r\geq   \frac{8{s_2c_2\rho}}{3s_1c_1{\epsilon^{'}}^2}\mathrm{ln}(\frac{4\rho}{\delta})$ with    $\epsilon^{'} \in (0,\frac{1}{2})$ and $\delta \in (0,1)$,   we have
		\begin{align*}
			\|V^TSS^TV-I\|_2\leq \epsilon^{'},
		\end{align*}
		with the probability at least $1-\delta$.	
	\end{lemma}
	\begin{proof}
		The proof  can be accomplished along the line of the proof of \cite[Theorem 3]{chowdhury2018iterative}.	However, for our case, 
		it is necessary to note that
		\begin{align}
			\|F_t\|_2&=\|M_tM^T_t-\frac{V^TV}{r}\|_2\leq \mathop{\rm{max}}{\{\|M_tM^T_t\|_2, \frac{1}{r}\}} \quad  \rm{by} \ Lemma \ \ref{lem.4} \nonumber\\
			&	=\frac{1}{r}\mathop{\rm{max}}\limits_{1\leq i \leq p}\{ \|\frac{(V^{i})^T}{\sqrt{\pi^{OPL}_{i}}}\frac{V^{i}}{\sqrt{\pi^{OPL}_{i}}}\|_2,1 \}  =\frac{1}{r}\mathop{\rm{max}}\limits_{1\leq i \leq p}\{ \frac{\|V^{i}\|^2_2
			}{\pi^{OPL}_i}, 1\}  
			\nonumber\\
			&=\frac{1}{r}\mathop{\rm{max}}\limits_{1\leq i \leq p}\{ \frac{\|V^{i}\|^2_2\sum_{i=1}^{p}	\mid \hat{\beta}_{rls(i)}	\mid \|A_i\|_2
			}{	\mid \hat{\beta}_{rls(i)}	\mid \|A_i\|_2}, 1\} \quad  \rm{by} \ \eqref{thm2.1}  \nonumber\\   
			&\leq\frac{1}{r}\mathop{\rm{max}}\limits_{1\leq i \leq p}\{ \frac{\|V^{i}\|^2_2\sum_{i=1}^{p}s_2c_2\|V^i\|^2_2
			}{s_1c_1\|V^i\|^2_2}, 1\} \quad  \rm{by} \ \eqref{lem.5.1}  \nonumber \\  
			&\leq\frac{1}{r}\mathop{\rm{max}}\limits_{1\leq i \leq p}\{\frac{s_2c_2}{s_1c_1}  {\sum_{i=1}^{p}\|V^i\|^2_2
			}, 1\}   \leq\frac{1}{r}\mathop{\rm{max}}\limits_{1\leq i \leq p}\{\frac{s_2c_2}{s_1c_1}\rho, 1\} \leq \frac{s_2c_2\rho}{rs_1c_1} \nonumber
		\end{align}
		and  
		\begin{align}
			\mathrm{E}(F^2_t) + \frac{(V^TV)^2}{r^2}&=\mathrm{E}(M_tM^T_t\|M_t\|^2_2) =\sum_{i=1}^{p}\pi^{OPL}_i\frac{(V^{i})^TV^{i}\|V^{i}\|^2_2}{r^2(\pi^{OPL}_i)^2}\nonumber\\
			&=\frac{1}{r^2}\sum_{i=1}^{p}\frac{(V^{i})^TV^{i}\|V^{i}\|^2_2\sum_{i=1}^{p}	\mid \hat{\beta}_{rls(i)}	\mid \|A_i\|_2}{	\mid \hat{\beta}_{rls(i)}	\mid \|A_i\|_2} \quad  \rm{by} \ \eqref{thm2.1} \nonumber\\    
			&\preccurlyeq\frac{1}{r^2}\sum_{i=1}^{p}\frac{(V^{i})^TV^{i}\|V^{i}\|^2_2\sum_{i=1}^{p}s_2c_2\|V^i\|^2_2
			}{s_1c_1\|V^i\|^2_2} \quad  \rm{by} \ \eqref{lem.5.1}  \nonumber  \\
			&\preccurlyeq\frac{s_2c_2}{r^2s_1c_1}\sum_{i=1}^{p}(V^{i})^TV^{i}\sum_{i=1}^{p}\|V^i\|^2_2\nonumber\\
			&=\frac{s_2c_2\rho}{r^2s_1c_1}\sum_{i=1}^{p}(V^{i})^TV^{i}=\frac{s_2c_2\rho}{r^2s_1c_1}I_\rho,\nonumber
		\end{align}
		where  $F_t=M_tM^T_t-\frac{V^TV}{r}$ with $M_t=\frac{(V^{i_t})^T}{\sqrt{r\pi^{OPL}_{i_t}}}$ and  $t=1,\cdots, r$.	
	\end{proof}

	\noindent ${\textbf{Proof  of Theorem \ref{thm.3}}}$.  
	Noting $\hat{\beta}=\frac{1}{\lambda}V\Sigma U^T\hat{z}$ and $\hat{\beta}_{rls}=\frac{1}{\lambda}V\Sigma U^T\hat{z}^{\ast}$,	 we can rewrite \eqref{thm3.0}
	as
	\begin{align}\label{proofthm3.0}
		\frac{1}{\lambda}\|\Sigma U^T(\hat{z}-\hat{z}^{\ast})\|_2	\leq\frac{\epsilon}{\lambda}\|\Sigma U^T\hat{z}^{\ast}\|_2.
	\end{align}
	To prove \eqref{proofthm3.0},  we  
	define the loss functions $L(z)$ and  $\hat{L}(z)$ as 
	\begin{align*}
		L(z)=\frac{1}{2\lambda}\|A^Tz\|^2_2 +\frac{1}{2}\|z\|^2_2-z^Ty 
	\end{align*}
	and 
	\begin{align*}	
		\mathit{	\hat{L}(z)=\frac{\rm{1}}{\rm{2}\lambda}\|S^TA^Tz\|^{\rm{2}}_{\rm{2}} +\frac{\rm{1}}{\rm{2}}\|z\|^{\rm{2}}_{\rm{2}}-z^Ty}.
	\end{align*} 
	Thus,	by  Taylor expansion,  we can acquire 
	\begin{align}\label{proofthm3.3}
		\hat{L}(\hat{z})=\hat{L}(\hat{z}^{\ast})+(\hat{z}-\hat{z}^{\ast})^T\triangledown\hat{L}(\hat{z}^{\ast})+(\hat{z}-\hat{z}^{\ast})^T\triangledown^2\hat{L}(z_0)(\hat{z}-\hat{z}^{\ast}),
	\end{align}
	where  
	$\hat{z}^{\ast}$ and $\hat{z}$ minimize the loss functions $L(z)$ and $\hat{L}(z)$, respectively, and $z_0\in [\hat{z},\hat{z}^{\ast}]$.
	Moreover, following
	$(\triangledown^2\hat{L}(z_0)-\triangledown^2{L}(z_0))\hat{z}^{\ast}=\triangledown\hat{L}(\hat{z}^{\ast})-\triangledown{L}(\hat{z}^{\ast})$, which is from
	\begin{align*}
		\triangledown\hat{L}(\hat{z}^{\ast})=(\frac{1}{\lambda}ASS^TA^T+I )\hat{z}^{\ast}-y, \  \triangledown{L}(\hat{z}^{\ast})=(\frac{1}{\lambda}AA^T+I )\hat{z}^{\ast}-y,
	\end{align*}
	and 
	\begin{align} \label{proofthm3.4.0.0}
		\triangledown^2\hat{L}(z_0)=\frac{1}{\lambda}ASS^TA^T+I, \ \triangledown^2{L}(z_0)=\frac{1}{\lambda}AA^T+I,
	\end{align}
	we can obtain
	\begin{align}
		&\hat{L}(\hat{z}^{\ast})+(\hat{z}-\hat{z}^{\ast})^T(\triangledown^2\hat{L}(z_0)-\triangledown^2{L}(z_0))\hat{z}^{\ast}=\hat{L}(\hat{z}^{\ast})+(\hat{z}-\hat{z}^{\ast})^T(\triangledown\hat{L}(\hat{z}^{\ast})-\triangledown{L}(\hat{z}^{\ast})).\nonumber
	\end{align}
	Thus,	considering that 
	\begin{align*}
		\hat{L}(\hat{z}^{\ast})+(\hat{z}-\hat{z}^{\ast})^T(\triangledown\hat{L}(\hat{z}^{\ast})-\triangledown{L}(\hat{z}^{\ast})) 	\leq	\hat{L}(\hat{z}^{\ast})+(\hat{z}-\hat{z}^{\ast})^T\triangledown\hat{L}(\hat{z}^{\ast}),
	\end{align*}			
	which is derived by  the fact  $(\hat{z}-\hat{z}^{\ast})^T\triangledown{L}(\hat{z}^{\ast})\geq0$, and 	noting \eqref{proofthm3.3}, we can gain 
	\begin{align*}
		&\hat{L}(\hat{z}^{\ast})+(\hat{z}-\hat{z}^{\ast})^T(\triangledown^2\hat{L}(z_0)-\triangledown^2{L}(z_0))\hat{z}^{\ast} \leq\hat{L}(\hat{z})-(\hat{z}-\hat{z}^{\ast})^T\triangledown^2\hat{L}(z_0)(\hat{z}-\hat{z}^{\ast}).   
	\end{align*}
	Further, by $\hat{L}(\hat{z}^{\ast})\geq\hat{L}(\hat{z})$,  we have
	\begin{align} 
		(\hat{z}-\hat{z}^{\ast})^T(\triangledown^2L(z_0)-\triangledown^2\hat{L}(z_0))\hat{z}^{\ast} \geq (\hat{z}-\hat{z}^{\ast})^T\triangledown^2\hat{L}(z_0)(\hat{z}-\hat{z}^{\ast}),\nonumber
	\end{align} 
	which together with
	\begin{align*}
		(\hat{z}-\hat{z}^{\ast})^T\triangledown^2\hat{L}(z_0)(\hat{z}-\hat{z}^{\ast})\geq(\hat{z}-\hat{z}^{\ast})^T\frac{1}{\lambda}ASS^TA^T(\hat{z}-\hat{z}^{\ast})
	\end{align*} 	
	and \eqref{proofthm3.4.0.0}  leads to
	\begin{align}			
		(\hat{z}-\hat{z}^{\ast})^T(\frac{1}{\lambda}AA^T-\frac{1}{\lambda}ASS^TA^T )\hat{z}^{\ast} \geq (\hat{z}-\hat{z}^{\ast})^T\frac{1}{\lambda}ASS^TA^T(\hat{z}-\hat{z}^{\ast}).  \nonumber 
	\end{align} 
	Thus, based on $A=U\Sigma V^T$, it is straightforward to get 
	\begin{align*}
		&	\frac{1}{\lambda^2}(\hat{z}-\hat{z}^{\ast})^T(U\Sigma^2U^T-U\Sigma V^T SS^TV\Sigma U^T )\hat{z}^{\ast} 
		\geq \frac{1}{\lambda^2}(\hat{z}-\hat{z}^{\ast})^T U\Sigma V^T SS^TV\Sigma U^T (\hat{z}-\hat{z}^{\ast}),
	\end{align*}
	which is also  allowed to be rewritten as 
	\begin{align}
		&	\frac{1}{\lambda^2}[\Sigma U^T(\hat{z}-\hat{z}^{\ast})]^T(I-V^T SS^TV)\Sigma U^T \hat{z}^{\ast} 
		\geq \frac{1}{\lambda^2}[\Sigma U^T(\hat{z}-\hat{z}^{\ast})]^T  V^T SS^TV [\Sigma U^T(\hat{z}-\hat{z}^{\ast})]. \label{proofthm3.6.0}
	\end{align}
	Adding $\frac{1}{\lambda^2}[\Sigma U^T(\hat{z}-\hat{z}^{\ast})]^T[\Sigma U^T(\hat{z}-\hat{z}^{\ast})]$ to both sides of \eqref{proofthm3.6.0} gives
	\begin{align}\label{proofthm3.6}
		&  \frac{1}{\lambda^2}[\Sigma U^T(\hat{z}-\hat{z}^{\ast})]^T(I-V^T SS^TV)\Sigma U^T \hat{z}^{\ast}	
		 + \frac{1}{\lambda^2}[\Sigma U^T(\hat{z}-\hat{z}^{\ast})]^T
		(I-V^T SS^TV)[\Sigma U^T(\hat{z}-\hat{z}^{\ast})]  
		\nonumber\\ 
		& \geq \frac{1}{\lambda^2}[\Sigma U^T(\hat{z}-\hat{z}^{\ast})]^T[\Sigma U^T(\hat{z}-\hat{z}^{\ast})]. 
	\end{align}
	Taking the Euclidean norm  on both  sides of \eqref{proofthm3.6}, we  obtain
	\begin{align*} 
		&\frac{1}{\lambda^2}\|\Sigma U^T(\hat{z}-\hat{z}^{\ast})\|_2\|I-V^T SS^TV\|_2\|\Sigma U^T \hat{z}^{\ast}\|_2
		\nonumber	\\ 	 
		&	+ \frac{1}{\lambda^2}\|\Sigma U^T(\hat{z}-\hat{z}^{\ast})\|_2\|I-V^T SS^TV\|_2\|\Sigma U^T(\hat{z}-\hat{z}^{\ast})\|_2 \nonumber	\\ 	 
		&\geq\frac{1}{\lambda^2}\|\Sigma U^T(\hat{z}-\hat{z}^{\ast})\|^2_2,
	\end{align*}
	which combined  with  Lemma  \ref{lem.5} 
	indicates that
	\begin{align}
		\frac{1}{\lambda}\epsilon^{'}\|\Sigma U^T \hat{z}^{\ast}\|_2	+ \frac{1}{\lambda}\epsilon^{'}\|\Sigma U^T(\hat{z}-\hat{z}^{\ast})\|_2	\geq\frac{1}{\lambda}\|\Sigma U^T(\hat{z}-\hat{z}^{\ast})\|_2. \label{proofthm3.7.1}
	\end{align}
	By rewriting  \eqref{proofthm3.7.1} as 
	\begin{align*}
		\frac{1}{\lambda}\|\Sigma U^T(\hat{z}-\hat{z}^{\ast})\|_2	\leq\frac{\epsilon^{'}}{1-\epsilon^{'}}\frac{1}{\lambda}\|\Sigma U^T\hat{z}^{\ast}\|_2
	\end{align*}
	and considering  the fact $\epsilon^{'}<\frac{1}{2}$, we have
	\begin{align*}
		\frac{1}{\lambda}\|\Sigma U^T(\hat{z}-\hat{z}^{\ast})\|_2	\leq\frac{2\epsilon^{'}}{\lambda}\|\Sigma U^T\hat{z}^{\ast}\|_2.
	\end{align*}
	Thus, setting $\epsilon=2\epsilon^{'}$,  we get \eqref{proofthm3.0}. That is,
	\eqref{thm3.0} 
	is arrived. 
	
\section{Proof of Theorem \ref{thm.4}}
	\label{supp.thm4}
	
	
	The proof  can be completed along the line of the proof of  
	Theorem 6 in \cite{chen2015fast}.
	However, when we bound $ \|R\|_2 $ with 
	\begin{align*}
		R=(\lambda\Sigma^{-1}+\Sigma)^{-1}\Sigma(V^TS^TSV-I),
	\end{align*}
	Lemma    \ref{lem.5}
	is adopted but not the oblivious subspace embedding theorem  \cite[Theorem 5]{chen2015fast}, namely,  
	\begin{align}
		\|R\|_2&\leq\|(\lambda\Sigma^{-1}+\Sigma)^{-1}\Sigma(V^TS^TSV-I)\|_2 \nonumber \\
		&\leq\|(\lambda\Sigma^{-1}+\Sigma)^{-1}\Sigma\|_2\|V^TS^TSV-I\|_2 \nonumber\\
		&\leq\epsilon^{'}\|(\lambda\Sigma^{-1}+\Sigma)^{-1}\Sigma\|_2  \quad \ \rm{by} \ Lemma \ \ref{lem.5} \nonumber\\  
		&\leq\epsilon^{'},\nonumber
	\end{align}
	where $ \epsilon^{'}  $ satisfies  $\epsilon^{'}=\frac{\epsilon}{2}$. 

	\section{Proof of Theorem \ref{thm.5}} 
	\label{app.D}

	The proof is similar to the  one of Theorem \ref{thm.1}  
	(see Appendix \ref{app.A}),  and we begin by presenting two lemmas. 

	\begin{lemma}\label{lem.6}
		Assume that the condition \eqref{rem.c2.2} and \eqref{thm5.0} hold. 
		Then, for $m=1$ and 	$\pi^{NOPL}_i $  in \eqref{sec4.1.0}, 
		we have	
		\begin{align}\label{lem6.1}
			\sum_{i=1}^{p}\pi^{NOPL}_i\|\frac{\widetilde{e}_i}{p}\|^{3}_2=O_p(1),
		\end{align}
		where $\widetilde{e}_i=(\frac{A_iA^T_i}{\pi^{NOPL}_i}+\lambda I)\hat{z}^{\ast}-\widetilde{y}$ with $\widetilde{y}=\lambda y$ and  $\hat{z}^{\ast}$ being as in \eqref{rls.2}.	
	\end{lemma}
	
	\begin{proof}
		Similar to  the proof of  Lemma \ref{lem.1},  based on  
		\eqref{rem.c2.2}, 
		\eqref{sec4.1.0},   \eqref{thm5.0}, \eqref{proof.lem.1.0}, and 
		\eqref{proof.lem.1.01}, 
		we have 
		
		\begin{align}
			\sum_{i=1}^{p}\pi^{NOPL}_i\|\frac{\widetilde{e}_i}{p}\|^{3}_2	
			\leq	&\frac{\|\widetilde{y}\|_2^{3}\sigma^{3}_{n}(AA^T+\lambda I)}{p^{3}}(\sum_{i=1}^{p}\frac{\|A_iA^T_i\|^3_2}{(\pi^{NOPL}_i)^{2}}+3\lambda\sum_{i=1}^{p}\frac{\|A_iA^T_i\|^2_2}{\pi^{NOPL}_i}
			\nonumber\\ 
			&+3\lambda^2\sum_{i=1}^{p}{\|A_iA^T_i\|_2}+\lambda^3)+  \frac{\|\widetilde{y}\|_2^{3}}{p^{3}}	\nonumber\\  &+3\frac{\|\widetilde{y}\|_2^{3}\sigma^{2}_{n}(AA^T+\lambda I)}{p^{3}}(\sum_{i=1}^{p}\frac{\|A_iA^T_i\|^2_2}{\pi^{NOPL}_i}
			+2\lambda\sum_{i=1}^{p}{\|A_iA^T_i\|_2}+\lambda^2)
			\nonumber \\ 
			&+3\frac{\|\widetilde{y}\|_2^{3}\sigma_{n}(AA^T+\lambda I)}{p^{3}}(\sum_{i=1}^{p}{\|A_iA^T_i\|_2}+\lambda) \nonumber \\	
			=&\frac{\|\widetilde{y}\|_2^{3}\sigma^{3}_{n}(AA^T+\lambda I)}{p^{3}}[\sum_{i=1}^{p}\frac{\|A_iA^T_i\|^3_2}{(	\mid \widetilde{\beta}_{(i)}	\mid \|A_i\|_2)^{2}}(\sum_{i=1}^{p}	\mid \widetilde{\beta}_{(i)}	\mid \|A_i\|_2)^2   \nonumber \\
			&+3\lambda\sum_{i=1}^{p}\frac{\|A_iA^T_i\|^2_2}{	\mid \widetilde{\beta}_{(i)}	\mid \|A_i\|_2}\sum_{i=1}^{p}	\mid \widetilde{\beta}_{(i)}	\mid \|A_i\|_2
			+3\lambda^2\sum_{i=1}^{p}{\|A_iA^T_i\|_2	}
			\nonumber \\
			&+\lambda^3]+  \frac{\|\widetilde{y}\|_2^{3}}{p^{3}} +3\frac{\|\widetilde{y}\|_2^{3}\sigma^{2}_{n}(AA^T+\lambda I)}{p^{3}}(\sum_{i=1}^{p}\frac{\|A_iA^T_i\|^2_2}{	\mid \widetilde{\beta}_{(i)}	\mid \|A_i\|_2}
			\nonumber \\
			&	\sum_{i=1}^{p}	\mid \widetilde{\beta}_{(i)}	\mid \|A_i\|_2
			+2\lambda\sum_{i=1}^{p}{\|A_iA^T_i\|_2}+\lambda^2)
			\nonumber \\
			&	+3\frac{\|\widetilde{y}\|_2^{3}\sigma_{n}(AA^T+\lambda I)}{p^{3}}(\sum_{i=1}^{p}{\|A_iA^T_i\|_2}+\lambda) \quad \rm{by} \   \eqref{sec4.1.0} \nonumber \\	
			&\leq\frac{\|\widetilde{y}\|_2^{3}\sigma^{3}_{n}(AA^T+\lambda I)}{p^{3}}[\frac{N^2_2}{N^2_1}(\sum_{i=1}^{p}\|A_i\|^2_2)^3+3\lambda\frac{N_2}{N_1}(\sum_{i=1}^{p}{\|A_i\|^2_2})^2
			\nonumber \\
			&+3\lambda^2\sum_{i=1}^{p}{\|A_i\|^2_2}+\lambda^3]
			+  \frac{\|\widetilde{y}\|_2^{3}}{p^{3}}  	\nonumber\\ &+3\frac{\|\widetilde{y}\|_2^{3}\sigma^{2}_{n}(AA^T+\lambda I)}{p^{3}}[\frac{N_2}{N_1}(\sum_{i=1}^{p}{\|A_i\|^2_2})^2
			+2\lambda\sum_{i=1}^{p}{\|A_i\|^2_2}+\lambda^2]
			\nonumber\\ &+3\frac{\|\widetilde{y}\|_2^{3}\sigma_{n}(AA^T+\lambda I)}{p^{3}}(\sum_{i=1}^{p}{\|A_i\|^2_2}+\lambda) \quad \rm{by}   \ \eqref{thm5.0} \nonumber \\							
			=&O_p(1),  \quad \rm{by} 
			\ \eqref{proof.lem.1.01} \ and   \ \eqref{rem.c2.2}  \nonumber
		\end{align}
		where 
		the first inequality  is gained by replacing $\pi_i$ in \eqref{proof.lem.1.0} with  $\pi^{NOPL}_i$. Then, \eqref{lem6.1} is obtained.
	\end{proof}
	
	\begin{lemma}\label{lem.7}
		To the assumption of    Lemma \ref{lem.6}, add that the condition \eqref{rem.c2.1} holds. 
		Then, for  $m=1$ and  $\pi^{NOPL}_i $   in \eqref{sec4.1.0}, 
		conditional on $\mathcal{F}_n$ and $\widetilde{\beta}$ in probability,  
		we have
		\begin{align}
			\frac{\widetilde{M}_A-M_A}{p}=O_{p	\mid \mathcal{F}_n}(r^{-1/2}), \label{lem7.1}\\
			\frac{\widetilde{e}^{\ast}}{p}=O_{p	\mid \mathcal{F}_n}(r^{-1/2}), \label{lem7.2}
		\end{align}
		where $M_A=AA^T+\lambda I$, $\widetilde{M}_A=A\widetilde{S}\widetilde{S}^TA^T+\lambda I$ with $  \widetilde{S} $ constructed by $\pi^{NOPL}_i$, 
		and $\widetilde{e}^{\ast}=(\widetilde{M}_A\hat{z}^{\ast}-\widetilde{y})$ with $\widetilde{y}=\lambda y$ and  $\hat{z}^{\ast}$ being as in \eqref{rls.2}.
	\end{lemma}
	
	\begin{proof}
		The proof can be completed similar to the proof of Lemma \ref{lem.2}.  We only need to replace $\pi_i$ with $\pi_i^{NOPL}$,  and note that
		\begin{align}
			\frac{1}{p^2}\sum_{i=1}^{p}\frac{A_{i}A_{i}^TA_{i}A_{i}^T}{\pi_{i}^{NOPL}}&=\frac{1}{p^2}(\sum_{i=1}^{p}\frac{A_{i}A_{i}^TA_{i}A_{i}^T}{	\mid \widetilde{\beta}_{(i)}	\mid \|A_i\|_2})(\sum_{i=1}^{p}	\mid \widetilde{\beta}_{(i)}	\mid \|A_i\|_2)
			\nonumber \\
			&	\leq\frac{N_2}{N_1p^2}(\sum_{i=1}^{p}\frac{A_{i}A_{i}^TA_{i}A_{i}^T}{\|A_i\|^2_2})(\sum_{i=1}^{p}\|A_i\|^2_2)\nonumber \\
			&=\frac{N_2}{N_1p^2}(\sum_{i=1}^{p}{A_{i}A_{i}^T})(\sum_{i=1}^{p}\|A_i\|^2_2)\nonumber \\
			& =O_p(1), \quad \rm{by} \ \eqref{rem.c2.2} \ and \ \eqref{rem.c2.1}  \label{lemrem7.2} 
		\\
			\mathrm{E}(\widetilde{M}_A	\mid  A)&=\mathrm{E}_{\widetilde{\beta}}[\mathrm{E}(\widetilde{M}_A	\mid  A,\widetilde{\beta})],\nonumber\\
			\mathrm{Var}[\frac{(\widetilde{M}_A-M_A)\ell}{p}	\mid  A]&=\mathrm{E}_{\widetilde{\beta}}\{	\mathrm{Var}[\frac{(\widetilde{M}_A-M_A)\ell}{p}	\mid  A,\widetilde{\beta}]\},\nonumber\\
			\mathrm{E}(\widetilde{e}_{i_t}	\mid \mathcal{F}_n)&=\mathrm{E}_{\widetilde{\beta}}[\mathrm{E}(\widetilde{e}_{i_t}	\mid \mathcal{F}_n,\widetilde{\beta})],\nonumber\\
			\mathrm{Var}(\frac{\ell^T \widetilde{e}^{\ast}}{p}	\mid \mathcal{F}_n)&=\mathrm{E}_{\widetilde{\beta}}[		\mathrm{Var}(\frac{\ell^T \widetilde{e}^{\ast}}{p}	\mid \mathcal{F}_n,\widetilde{\beta})],\nonumber
		\end{align} 
		where	$\mathrm{E}_{\widetilde{\beta}}$ denotes the expectation on $\widetilde{\beta}$.

	\end{proof}

	\begin{remark}\label{lemrem7.1}
		
		The results
		\eqref{lem7.1} and \eqref{lem7.2} still hold 
		when $\widetilde{M}_A=A{S}^{\ast}{{S}^{\ast}}^TA^T+\lambda I$  with $ {S}^{\ast} \in \mathbb{R}^{p \times r_0} $ formed by $\pi^{COL}_i$.
		
	\end{remark}

	\begin{corollary}\label{cor.1}
		For   $ {S}^{\ast} \in \mathbb{R}^{p \times r_0} $ formed by $\pi^{COL}_i$,  $\widetilde{z}=(A{S}^{\ast}{{S}^{\ast}}^TA^T+\lambda I)^{-1}\widetilde{y}$ constructed by  Algorithm 2 
		satisfies 
		\begin{align}\label{lemrem7.1.2}
			\|\widetilde{z}-\hat{z}^{\ast}\|_2=O_{p	\mid \mathcal{F}_n}(r_0^{-1/2}).
		\end{align}
		
	\end{corollary}
	
	\begin{proof}
		Similar to   \eqref{proofthem.1.0.0},  considering \eqref{rem.c2.1} and Remark \ref{lemrem7.1},  we can get 
		\begin{align*}
			\widetilde{z}-\hat{z}^{\ast}=-(\frac{M_A}{p})^{-1}\frac{\widetilde{e}^{\ast}}{p}+(\frac{\widetilde{M}_A}{p})^{-1}(\frac{\widetilde{M}_A-M_A}{p})(\frac{M_A}{p})^{-1}\frac{\widetilde{e}^{\ast}}{p}=O_{p	\mid \mathcal{F}_n}(r_0^{-1/2}),
		\end{align*} 
		which suggests that \eqref{lemrem7.1.2}	holds.
	\end{proof}
	
	\noindent	${\textbf{Proof  of Theorem \ref{thm.5}}}$. 
	Similar to the proof of Theorem \ref{thm.1}
	, noting  \eqref{rem.c2.2}, \eqref{rem.c2.1}, \eqref{lemrem7.2}, and Lemmas \ref{lem.6} and \ref{lem.7}, and replacing  $\pi_i$ and ${e}_{i_t}$ in the proof of Theorem \ref{thm.1} 
	with $\pi_i^{NOPL}$ and $\widetilde{e}_{i_t}$, respectively, we first get
	\begin{align}
		\hat{z}_1-\hat{z}^{\ast}
		=-(\frac{M_A}{p})^{-1}\frac{\widetilde{e}^{\ast}}{p}+O_{p	\mid \mathcal{F}_n}(r^{-1})\label{proofthem.5.0},\\
		(\frac{\widetilde{V}_c}{r})^{-1/2}(\frac{\widetilde{e}_{\ast}}{p})\xrightarrow{L}N(0,I),  \nonumber 
	\end{align}
	where 
	\begin{align*}
		\hat{z}_1=(A\widetilde{S}\widetilde{S}^TA^T+\lambda I)^{-1}\widetilde{y}=\widetilde{M}^{-1}_A\widetilde{y}, 
		\\
		\mathit	{\widetilde{V}_c=\sum_{i=1}^{p}\frac{A_{i}A_{i}^T\hat{z}^{\ast}\hat{z}^{\ast T}A_{i}A_{i}^T}{p^{\rm{2}}\pi^{NOPL}_{i}}=O_p}(1).
	\end{align*}
	To  get \eqref{thm5.1}, 
	in the following,   we  need to further prove  	\begin{align} \label{proofthm5.6.0}
		{V}^{-1/2}_{OPL}(\hat{z}_1-\hat{z}^{\ast})=-{V}^{-1/2}_{OPL}(\frac{M_A}{p})^{-1}(\frac{\widetilde{V}_c}{r})^{1/2}(\frac{\widetilde{V}_c}{r})^{-1/2}\frac{\widetilde{e}^{\ast}}{p}+O_{p	\mid \mathcal{F}_n}(r^{-1/2}),
	\end{align} where $ {V}^{-1/2}_{OPL}(\frac{M_A}{p})^{-1}(\frac{\widetilde{V}_c}{r})^{1/2}$ satisfies	\begin{align} \label{proofthm5.6.1}
		{V}^{-1/2}_{OPL}(\frac{M_A}{p})^{-1}(\frac{\widetilde{V}_c}{r})^{1/2}[{V}^{-1/2}_{OPL}(\frac{M_A}{p})^{-1}(\frac{\widetilde{V}_c}{r})^{1/2}]^T
		=I+ O_{p	\mid \mathcal{F}_n}(r_0^{-1/2}).
	\end{align}

	Considering \eqref{rem.c2.2}, \eqref{rem.c2.1}, 
	\eqref{thm2.1} and \eqref{thm5.0}, 
	we first obtain
	\begin{align}
		\frac{1}{p^2}\sum_{i=1}^{p}\frac{A_{i}A_{i}^TA_{i}A_{i}^T}{\pi_{i}^{OPL}}&=\frac{1}{p^2}(\sum_{i=1}^{p}\frac{A_{i}A_{i}^TA_{i}A_{i}^T}{	\mid \hat{\beta}_{rls(i)}	\mid \|A_i\|_2})(\sum_{i=1}^{p}	\mid \hat{\beta}_{rls(i)}	\mid \|A_i\|_2) \quad \rm{by} \ \eqref{thm2.1} \nonumber\\  
		&\leq\frac{N_4}{N_3p^2}(\sum_{i=1}^{p}\frac{A_{i}A_{i}^TA_{i}A_{i}^T}{\|A_i\|^2_2})(\sum_{i=1}^{p}\|A_i\|^2_2) \quad \rm{by} \  \eqref{thm5.0} \nonumber \\ 
		&=\frac{N_4}{N_3p^2}(\sum_{i=1}^{p}{A_{i}A_{i}^T})(\sum_{i=1}^{p}\|A_i\|^2_2) 
		\nonumber \\
		&=O_p(1), \quad \rm{by} \ \eqref{rem.c2.2} \ and \ \eqref{rem.c2.1} \nonumber 
	\end{align}
	which indicates 
	\begin{align}\label{proofthm5.5} 
		V_{cOPL}=\sum_{i=1}^{p}\frac{A_iA^T_i\hat{z}^{\ast}\hat{z}^{\ast T}A_iA^T_i}{p^2\pi^{OPL}_i}=O_p(1).
	\end{align}
	From  \eqref{rem.c2.1} and \eqref{proofthm5.5}, it is evident to get 
	\begin{align} \label{proofthm5.6}
		{V}_{OPL}=(\frac{M_A}{p})^{-1}\frac{V_{cOPL}}{r}(\frac{M_A}{p})^{-1}=O_p(r^{-1}),
	\end{align}
	which  combined  with \eqref{proofthem.5.0}  suggests that \eqref{proofthm5.6.0} holds, that is,
	\begin{align*} 
		{V}^{-1/2}_{OPL}(\hat{z}_1-\hat{z}^{\ast})&=-{V}^{-1/2}_{OPL}(\frac{M_A}{p})^{-1}\frac{\widetilde{e}^{\ast}}{p}+O_{p	\mid \mathcal{F}_n}(r^{-1/2})\nonumber\\
		&=-{V}^{-1/2}_{OPL}(\frac{M_A}{p})^{-1}(\frac{\widetilde{V}_c}{r})^{1/2}(\frac{\widetilde{V}_c}{r})^{-1/2}\frac{\widetilde{e}^{\ast}}{p}+O_{p	\mid \mathcal{F}_n}(r^{-1/2}).
	\end{align*}
	Now, we need to demonstrate that  \eqref{proofthm5.6.1} also holds. Evidently, it suffices to show that  
	\begin{align} \label{proofthm5.6.0.1}
		{V}^{-1/2}_{OPL}(\frac{M_A}{p})^{-1}\frac{\widetilde{V}_c-V_{cOPL}}{r}(\frac{M_A}{p})^{-1}{V}^{-1/2}_{OPL}=O_{p	\mid \mathcal{F}_n}(r_0^{-1/2}),
	\end{align}  because
	\begin{align*}
		& {V}^{-1/2}_{OPL}(\frac{M_A}{p})^{-1}(\frac{\widetilde{V}_c}{r})^{1/2}[{V}^{-1/2}_{OPL}(\frac{M_A}{p})^{-1}(\frac{\widetilde{V}_c}{r})^{1/2}]^T\nonumber \\
		&={V}^{-1/2}_{OPL}(\frac{M_A}{p})^{-1}\frac{\widetilde{V}_c}{r}(\frac{M_A}{p})^{-1}{V}^{-1/2}_{OPL}\nonumber \\
		&={V}^{-1/2}_{OPL}(\frac{M_A}{p})^{-1}\frac{V_{cOPL}}{r}(\frac{M_A}{p})^{-1}{V}^{-1/2}_{OPL}
		\nonumber \\
		&+{V}^{-1/2}_{OPL}(\frac{M_A}{p})^{-1}\frac{\widetilde{V}_c-V_{cOPL}}{r}(\frac{M_A}{p})^{-1}{V}^{-1/2}_{OPL}\nonumber \\
		&=I+{V}^{-1/2}_{OPL}(\frac{M_A}{p})^{-1}\frac{\widetilde{V}_c-V_{cOPL}}{r}(\frac{M_A}{p})^{-1}{V}^{-1/2}_{OPL}. 
	\end{align*}
	Noting 
	\begin{align*}
		\widetilde{V}_c=\underbrace{[\frac{1}{p}{(\sum_{i=1}^{p}\frac{A_{i}A_{i}^T\hat{z}^{\ast}\hat{z}^{\ast T}A_{i}A_{i}^T}{	\mid \widetilde{\beta}_{(i)}	\mid \|A_i\|_2})}]}_{\Phi_1}\underbrace{[\frac{1}{p}(\sum_{i=1}^{p}	\mid \widetilde{\beta}_{(i)}	\mid \|A_i\|_2)]}_{\Phi_2}, \\
		V_{cOPL}=\underbrace{[\frac{1}{p}(\sum_{i=1}^{p}\frac{A_{i}A_{i}^T\hat{z}^{\ast}\hat{z}^{\ast T}A_{i}A_{i}^T}{	\mid \hat{\beta}_{rls(i)}	\mid \|A_i\|_2})]}_{\Phi_3}\underbrace{[\frac{1}{p}(\sum_{i=1}^{p}	\mid \hat{\beta}_{rls(i)}	\mid \|A_i\|_2)]}_{\Phi_4},
	\end{align*}
	and  the basic  triangle inequality, 
	we  gain
	\begin{align*}
		\|\widetilde{V}_c-V_{cOPL}\|_2&=\|\Phi_1\Phi_2-\Phi_3\Phi_4\|_2
		\\
		&\leq\|\Phi_1-\Phi_3\|_2\|\Phi_2\|_2+\|\Phi_2-\Phi_4\|_2\|\Phi_3\|_2.
	\end{align*}
	Following \eqref{rem.c2.2}, \eqref{thm5.0}, \eqref{proof.lem.1.01},  and \eqref{lemrem7.1.2}, 
	it is evident to gain
	
	\begin{align*}
		\|\Phi_1-\Phi_3\|_2&\leq\|\frac{1}{p}\sum_{i=1}^{p}\frac{A_{i}A_{i}^T\hat{z}^{\ast}\hat{z}^{\ast T}A_{i}A_{i}^T}{\|A_i\|_2}(\frac{1}{	\mid \widetilde{\beta}_{(i)}	\mid }-\frac{1}{	\mid \hat{\beta}_{rls(i)}	\mid })\|_2 \\ &\leq\frac{1}{p}\sum_{i=1}^{p}\frac{\lambda^2\hat{\beta}_{rls(i)}^2\|A_{i}\|^2_2}{\|A_i\|_2}(\frac{	\mid \widetilde{\beta}_{(i)}-\hat{\beta}_{rls(i)}	\mid }{	\mid \hat{\beta}_{rls(i)}	\mid 	\mid \widetilde{\beta}_{(i)}	\mid })\\
		& \leq\frac{\lambda N_4}{pN_1}\sum_{i=1}^{p}\frac{\|A_{i}\|^3_2}{\|A_i\|^2_2}({\|A_i\|_2\|\widetilde{z}-\hat{z}^{\ast}\|_2}) \quad \rm{by} \ \eqref{thm5.0}\\ 
		& ={\|\widetilde{z}-\hat{z}^{\ast}\|_2}\sum_{i=1}^{p}{\frac{\lambda N_4\|A_{i}\|^2_2}{pN_1}}=O_{p	\mid \mathcal{F}_n}(r_0^{-1/2}), \quad \rm{by} \ \eqref{rem.c2.2} \ and \ \eqref{lemrem7.1.2} \\ 
	\|\Phi_2\|_2&\leq\frac{N_2\|y\|_2}{p}\sum_{i=1}^p{\|A_i\|^2_2}=O_p(1). \quad \rm{by} \ \eqref{rem.c2.2}, \  \eqref{thm5.0},\ and \ \eqref{proof.lem.1.01}  
	\end{align*}
	Similarly, we have $\|\Phi_2-\Phi_4\|_2=O_{p	\mid \mathcal{F}_n}(r_0^{-1/2})$ and $\|\Phi_3\|_2=O_p(1)$. Therefore, we get 
	\begin{align*}
		\|\widetilde{V}-V_{cOPL}\|_2=O_{p	\mid \mathcal{F}_n}(r_0^{-1/2}),
	\end{align*}
	which	combined with  \eqref{rem.c2.1} and  \eqref{proofthm5.6}
	yields  \eqref{proofthm5.6.0.1}.
	Putting the above discussions  and  the Slutsky's Theorem together,
	the result 
	\eqref{thm5.1}
	follows. 
	
	\section{Proof of Theorem \ref{thm.6}}
	\label{supp.thm6}

	Before providing the proof of Theorem \ref{thm.6}, we first present a lemma.
	
	\begin{lemma}\label{lem.8}
		To the assumption of Lemma \ref{lem.5},
		add that \eqref{rem4.3.0.1} holds
		and $r\geq   \frac{32{s_4c_2\rho}}{3s_3c_1{\epsilon}^2}\mathrm{ln}(\frac{4\rho}{\delta})$ with $\epsilon,\delta \in (0,1)$.
		Then, for any $\epsilon $, $\hat{w}_t$ obtained from the $t$-th iteration of Algorithm  2
		satisfies
		\begin{align}\label{sec4.4}
			\|\frac{A^T\hat{w}_t}{\lambda}-\frac{A^Tw^{\ast}_t}{\lambda}\|_2\leq\epsilon\|\frac{A^Tw^{\ast}_t}{\lambda}\|_2,
		\end{align}
		where $w^{\ast}_t$ is the solution of 
		\begin{align*}
			\mathop{\rm{min}}\limits_{w_t}\frac{1}{2\lambda}\|A^Tw_t\|_2^2+\frac{1}{2}\|w_t\|_2^2-w_t^Tb_t.  
		\end{align*}	
	\end{lemma}	
	
	\begin{proof}
		The proof can be completed along the line of the proof of  Theorem \ref{thm.3}. 
		Particularly, in this case,  Lemma  \ref{lem.5}
		still holds for $S=\widetilde{S}$, where $\widetilde{S}$ is formed by $\pi_{i}^{NOPL}$.
	\end{proof}

	\noindent $\textbf{Proof  of Theorem  \ref{thm.6}}$. 
	At the $t$-th iteration,
	following the discussion in  Remark \ref{rem4.3} 
	and \eqref{sec4.4}, and setting
	\begin{align*}
		\bigtriangleup^{\ast}_t=\frac{A^T{{w}^{\ast}_t}}{\lambda}=\frac{A^T{\hat{z}^{\ast}}}{\lambda}-\frac{A^T{\hat{z}_{t-1}}}{\lambda}
	\end{align*}
	and $\hat{\bigtriangleup}_t=\frac{A^T{\hat{w}_t}}{\lambda}$ as the estimator of $\bigtriangleup^{\ast}_t$, we can have 
	\begin{align}
		\|\hat{\bigtriangleup}_t-\bigtriangleup^{\ast}_t\|_2&\leq\epsilon\|\bigtriangleup^{\ast}_t\|_2 \quad \ \rm{by} \ \eqref{sec4.4} \nonumber 
		\\
		&=\epsilon\|\frac{A^T{\hat{z}^{\ast}}}{\lambda}-\frac{A^T{\hat{z}_{t-1}}}{\lambda}\|_2 \nonumber \\ 
		&=\epsilon\|\frac{A^T({\hat{z}_{t-2}}+{w_{t-1}^{\ast}})}{\lambda}-\frac{A^T({\hat{z}_{t-2}}+{\hat{w}_{t-1}})}{\lambda}\|_2
		\nonumber \\
		&\leq\epsilon\|\hat{\bigtriangleup}_{t-1}-\bigtriangleup^{\ast}_{t-1}\|_2
		\leq\epsilon^2\|\bigtriangleup^{\ast}_{t-1}\|_2.\nonumber
	\end{align} 
	As a result, 
	\begin{align*}
		\|\hat{\bigtriangleup}_m-\bigtriangleup^{\ast}_m\|_2&\leq\epsilon\|\hat{\bigtriangleup}_{m-1}-\bigtriangleup^{\ast}_{m-1}\|_2\leq\epsilon^m\|\bigtriangleup^{\ast}_1\|_2 \nonumber \\
		&\leq\epsilon^m\|\frac{A^T{\hat{z}^{\ast}}}{\lambda}-\frac{A^T{\hat{z}_{0}}}{\lambda}\|_2\nonumber \\
		&=\epsilon^m\|\frac{A^T{\hat{z}^{\ast}}}{\lambda}\|_2=\epsilon^m\|\hat{\beta}_{rls}\|_2. 
	\end{align*}
	Considering that $\hat{\beta}_m-\hat{\beta}_{rls}=\hat{\bigtriangleup}_m-\bigtriangleup^{\ast}_m$, the conclusion 
	is arrived.  
	


\end{appendices}

\bibliography{column_sampling_referrence}

\end{document}